\documentclass[smallcondensed,numbook,envcountsec,envcountsame]{svjour3}

\usepackage{microtype,amsmath}

\usepackage{tikz}
\usepackage{tikz-3dplot}
\usetikzlibrary{arrows,calc,shapes,decorations.pathreplacing}
\usepackage{float}
\usepackage{graphicx}
\usepackage{caption}
\usepackage{subcaption}

\usepackage{epsfig,amssymb} 

\newcommand {\N}{{\rm{I\!N}}}
\newcommand{\R}{\mbox{$\mathbb R$}}

\begin{document}
\title{Cut-off theorems for the $PV$-model.}

\author{Lisbeth Fajstrup}
\institute{Department of Mathematical Sciences, Aalborg University, Skjernvej 4A, 9220 Aalborg {\O}st, Denmark.
\email{fajstrup@math.aau.dk}}

\maketitle

\begin{abstract}
  For a $PV$ thread $T$ which accesses a set $\mathcal{R}$ of resources, each with a maximal capacity $\kappa:\mathcal{R}\to\mathbb{N}$, the PV-program $T^n$, where $n$ copies of $T$ are run in parallel, is deadlock free for all $n$ if and only if $T^M$ is deadlock free where $M$ is the sum of the capacities of the shared resources $M=\Sigma_{r\in\mathcal{R}}\kappa(r)$. This is a sharp bound: For all $\kappa:\mathcal{R}\to\mathbb{N}$ and finite $\mathcal{R}$ there is a thread $T$ using these resources such that $T^M$ has a deadlock, but $T^n$ does not for $n<M$.\\ Moreover, we prove a more general theorem for a set of different threads sharing resources $\mathcal{R}$: There are no deadlocks in $p=T1|T2|\cdots |Tn$ if and only if there are no deadlocks in $T_{i_1}|T_{i_2}|\cdots |T_{i_M}$ for any $M$-element subset $\{i_1,\ldots,i_M\}\subset [1:n]$.\\ For $\kappa(r)\equiv 1$, $T^n$ is serializable, i.e., all executions are equivalent to serial executions, for all $n$ if and only if $T^2$ is serializable. For general capacities, we define local obstructions to serializability - if no such obstruction exists, the program is serializable. There is no local obstruction to serializability in $T^n$ for all $n$ if and only if there is no local obstruction to serializability in $T^M$ for $M=\Sigma_{r\in\mathcal{R}}\kappa(r)+1$. The obstructions may be found using a deadlock algorithm in $T^{M+1}$. There is a generalization to  $p=T1|T2|\cdots |Tn$: If there are no local obstructions to serializability in any of the  sub programs, $T_{i_1}|T_{i_2}|\cdots |T_{i_M}$,  then $p$ is serializable. 
\end{abstract}

\maketitle

\section{Introduction}
We address the following: Verify properties of parallel programs in a setting, where users decide how many threads are run. In the case where an unknown number, $n$, of copies of the same thread $T$ may be run in parallel with itself as a program $T^n$, a cut-off result states that some property  holds for all $T^n$ if and only if it holds for $T^M$ for a fixed $M$ called the cut-off. Hence verification is required for only that  case. This is a very simple instance of parameterized verification in the sense of e.g. \cite{AbdullaDelzanno} - with $n$ as the parameter, verification is only needed for $n=M$.

Here, $T$ is a $PV$-program, it locks and releases resources from a set of resources $\mathcal{R}$, with a capacity function $\kappa:\mathcal{R}\to\mathbb{N}$ which gives an upper bound for how many of the parallel threads may hold a lock on the resource at a time.  The properties of $T^n$ investigated here are deadlocks and serializability.

To give the intuition and also the geometric interpretation of $PV$-programs in the sense of \cite{D68} and \cite{CR87}, we will consider some examples here. The formal definitions are in section 2. A thread $T1=Pa.Pb.Vb.Va.Pc.Vc$ will acces the resources $a,b,c$  and set a lock on them as follows: $Pr$ is a request to lock resource $r$. In a parallel program, the resource may be locked by other threads and $T$ may have to wait. If the lock is granted, $T1$ proceeds with a lock on $r$. $Vr$ is release of resource $r$. Hence, execution of this thread means lock $a$, lock $b$, release $b$, release $a$, lock $c$, release $c$. These are not atomic actions. They are only recording the synchronization of threads which share resources. See section 2 for definitions of valid threads and valid $PV$-programs.

The $PV$ formalism stems from E.~W.~ Dijkstra around 1963  \cite{D64}, where $P$ corresponds to the Dutch word Probeer - try (to access a resource), once access is obtained, lower a value by $1$ and $V$ is Verhoog, increase - which refers to increasing a value by $1$. In our notation, the value is  the number of threads which currently have access to the given resource, and $P$ increases it by one while $V$ decreases it by $1$. This is of course just a sign convention. 

The \emph{geometric model} of a thread $T=w_1.w_2.\ldots .w_l$ of \emph{length} $l$  is the interval $[0,l+1]$ and an execution of $T$ is a continuous, surjective, non-decreasing  map $\gamma: [0,1]\to [0,l+1]$ - thus allowing for all possible subdivisions of the non atomic actions - all points in the interval $[0,l+1]$ are considered states of $T$. The domain of $\gamma$, the interval $[0,1]$, is a choice which is not essential for the properties discussed here. 

The \emph{geometric model} of a parallel program $T1|T2|\cdots|Tn$ where $Ti$ has length $l_i$ is a subset $X$ of the hyperrectangle $[0,l_1+1]\times [0,l_2+1]\times\ldots [0,l_n+1]\subset \mathbb{R}^n$ comprised by all states $(x_1,\ldots,x_n)$ for which resources are at most locked to their capacity. See Fig.\ref{fig:firstPV} and Fig. \ref{fig:serialize}. The origo is denoted $\perp$ and the point $(l_1+1,l_2+1,\ldots,l_n+1)$ is denoted $\top$. 

\emph{Executions:} An execution of $T1|T2|\cdots|Tn$ is a continuous map $\gamma:[0,1]\to X$, $\gamma(t)=(\gamma_1(t),\ldots,\gamma_n(t))$ such that $\gamma(0)=\perp$ and $\gamma(1)=\top$ and such that $t,s\in [0,1], t\leq s$ implies $\gamma_i(t)\leq\gamma_i(s)$ for all $i$. $\gamma$ is an \emph{execution path}. Fig.~\ref{fig:firstPV} b), Fig.~\ref{fig:serialize} a). Executions are hence not just interleavings - see Fig.~\ref{fig:firstPV} b). The execution corresponding to the dotted black path shows $T1$ executing  $Pa.Pb.Vb.Va$ while $T2$ executes $Pb$ - it is truly concurrent. 

\emph{Deadlock:} Let $T2=Pb.Pa.Va.Vb.Pc.Vc$ and consider the parallel program $p=T1|T2$. If  $a$ and $b$ have capacity $1$, there is a deadlock:  The  state $(Pb,Pa)$ where $T1$ requests $b$ and $T2$ request $a$ is a  \emph{deadlock state}, since $a$ is locked by $T1$ and $b$ is locked by $T2$. Neither can proceed and nor can any concurrent execution Lem.~\ref{lem:deadlock}. This state is also denoted $(2,2)$ since it is the second $PV$ action of both threads. See Fig.~\ref{fig:firstPV} a).

For deadlocks, our cut-off theorem, Thm.~\ref{thm:deadgeneral} and Cor.~\ref{cor:deadsymmetric} states:  $T^n$ is deadlock free for all $n$ if and only if $T^M$ is deadlock free for $M=\Sigma_{r\in\mathcal{R}}\kappa(r)$. More generally: A program $p=T1|T2|\cdots |Tn$ is deadlock free if and only if all  sub programs $T_{i_1}|T_{i_2}|\cdots |T_{i_M}$ of $M$ of the threads comprising $p$ are deadlock free.

\begin{figure}[h]

\begin{subfigure}{.32\textwidth}
\begin{tikzpicture}[scale=0.4]
\draw [->] (0,0)--(7,0);
\draw [->] (0,0)--(0,7);

\filldraw[fill=white!30!blue,opacity=.8,draw=black] (1,2) rectangle (4,3);
\filldraw[fill=white!30!blue,opacity=.8,draw=black] (2,1) rectangle (3,4);
\filldraw[fill=white!30!blue,opacity=.8,draw=black] (5,5) rectangle (6,6);

\filldraw[fill=red] (2,2) circle (.08cm);
\draw (1,0) node[anchor=north]{$Pa$};
\draw (2,0) node[anchor=north]{$Pb$};
\draw (3,0) node[anchor=north]{$Vb$};
\draw (4,0) node[anchor=north]{$Va$};
\draw (5,0) node[anchor=north]{$Pc$};
\draw (6,0) node[anchor=north]{$Vc$};

\draw (0,1) node[anchor=east]{$Pb$};
\draw (0,2) node[anchor=east]{$Pa$};
\draw (0,3) node[anchor=east]{$Va$};
\draw (0,4) node[anchor=east]{$Vb$};
\draw (0,5) node[anchor=east]{$Pc$};
\draw (0,6) node[anchor=east]{$Vc$};

\draw (7,0) node[anchor=west]{$T1$};
\draw (0,7) node[anchor=south]{$T2$};
\foreach \x in {1,...,6}
\draw (\x cm,1pt)--(\x cm,-1pt);
\foreach \y in {1,...,6}
\draw (-1pt,\y cm)--(1pt,\y cm);

\end{tikzpicture}
\caption{The red dot indicates a deadlock - there is no way of proceeding. }
\end{subfigure}
\hspace{2mm}
\begin{subfigure}{.32\textwidth}
\begin{tikzpicture}[scale=0.4]
\draw [->] (0,0)--(7,0);
\draw [->] (0,0)--(0,7);

\filldraw[fill=white!30!blue,opacity=.8,draw=black] (1,2) rectangle (4,3);
\filldraw[fill=white!30!blue,opacity=.8,draw=black] (2,1) rectangle (3,4);
\filldraw[fill=white!30!blue,opacity=.8,draw=black] (5,5) rectangle (6,6);

\filldraw[fill=red] (2,2) circle (.08cm);
\draw (1,0) node[anchor=north]{$Pa$};
\draw (2,0) node[anchor=north]{$Pb$};
\draw (3,0) node[anchor=north]{$Vb$};
\draw (4,0) node[anchor=north]{$Va$};
\draw (5,0) node[anchor=north]{$Pc$};
\draw (6,0) node[anchor=north]{$Vc$};

\draw (0,1) node[anchor=east]{$Pb$};
\draw (0,2) node[anchor=east]{$Pa$};
\draw (0,3) node[anchor=east]{$Va$};
\draw (0,4) node[anchor=east]{$Vb$};
\draw (0,5) node[anchor=east]{$Pc$};
\draw (0,6) node[anchor=east]{$Vc$};

\draw (7,0) node[anchor=west]{$T1$};
\draw (0,7) node[anchor=south]{$T2$};
\foreach \x in {1,...,6}
\draw (\x cm,1pt)--(\x cm,-1pt);
\foreach \y in {1,...,6}
\draw (-1pt,\y cm)--(1pt,\y cm);
\draw[thick,dotted,rounded corners](0,0)--(4.5,1)--(4.5,6)--(7,7);
\draw[thick,dashed,rounded corners](0,0)--(1,4.5)--(6,4.5)--(7,7);
\draw[thick,dashed,green,rounded corners](0,0)--(0.1,7)--(7,7);
\draw[thick,dotted,green,rounded corners](0,0)--(7,0.1)--(7,7);
\end{tikzpicture}
\caption{ Four execution paths. }
\end{subfigure}
\begin{subfigure}{.32\textwidth}
\begin{tikzpicture}[scale=0.4]
\draw [->] (0,0)--(7,0);
\draw [->] (0,0)--(0,7);

\filldraw[fill=white!30!blue,opacity=.8,draw=black] (1,2) rectangle (4,3);
\filldraw[fill=white!30!blue,opacity=.8,draw=black] (2,1) rectangle (3,4);
\filldraw[fill=white!30!blue,opacity=.8,draw=black] (5,5) rectangle (6,6);

\filldraw[fill=red] (2,2) circle (.08cm);
\draw (1,0) node[anchor=north]{$Pa$};
\draw (2,0) node[anchor=north]{$Pb$};
\draw (3,0) node[anchor=north]{$Vb$};
\draw (4,0) node[anchor=north]{$Va$};
\draw (5,0) node[anchor=north]{$Pc$};
\draw (6,0) node[anchor=north]{$Vc$};

\draw (0,1) node[anchor=east]{$Pb$};
\draw (0,2) node[anchor=east]{$Pa$};
\draw (0,3) node[anchor=east]{$Va$};
\draw (0,4) node[anchor=east]{$Vb$};
\draw (0,5) node[anchor=east]{$Pc$};
\draw (0,6) node[anchor=east]{$Vc$};

\draw (7,0) node[anchor=west]{$T1$};
\draw (0,7) node[anchor=south]{$T2$};
\foreach \x in {1,...,6}
\draw (\x cm,1pt)--(\x cm,-1pt);
\foreach \y in {1,...,6}
\draw (-1pt,\y cm)--(1pt,\y cm);
\draw[thick,dotted,green,rounded corners](0,0)--(6,0.1)--(6,1)--(7,1)--(7,7);
\draw[thick,dotted,green,rounded corners](0,0)--(7,0.1)--(7,7);
\draw[thick,dotted,red,rounded corners](0,0)--(4.9,0.1)--(4.9,6.1)--(6.1,6.1)--(7,7);
\draw[thick,dotted,red,rounded corners](0,0)--(4.9,0.1)--(4.9,4.9)--(6.1,4.9)--(6.1,6.1)--(7,7);
\end{tikzpicture}
\caption{The two green execution paths are equivalent, the red executions are not.   }\label{fig:elemhtpy}
\end{subfigure}

\caption{  The geometric model of $Pa.Pb.Vb.Va.Pc.Vc|Pb.Pa.Va.Vb.Pc.Vc$ when all resources have capacity $1$. A point $(x_1,x_2)$ is a joint state of $T1$ and $T2$. The blue regions are states where more than one thread hold a lock on a resource - these are not allowed, when the resources have capacity $1$. If e.g. resource $c$ had capacity $2$ or more, the states in the upper right hand blue square would all be allowed. In b)  the two green paths are serial executions - one thread executes before the other. There are four equivalence classes of executions - represented by these four paths. In c) The green paths are equivalent - indication that $Vc^1.Pb^2$ is equivalent to $Pb^2.Vc^1$  all states in the square between $(6,0)$ and $(7,1)$ are allowed. See Ex.~\ref{ex:elemhtpy}. The two red execution paths are not. There is no homotopy between them and hence no directed homotopy. The blue "hole" of states not allowed prevents that. The red execution path to the right and up gives $T1$ access to $c$  before $T2$. Up and then right gives $T2$ acces first. }\label{fig:firstPV} 
\end{figure}

\begin{samepage}
\emph{Equivalence of executions}: Two execution paths $\gamma,\mu:[0,1]\to X$ are equivalent if there is a continuous map $H:[0,1]\times [0,1]\to X$ such that 
\begin{enumerate}
\item For all $s$, $H(0,s)=\perp$ and $H(1,s)=\top$.
\item For all $t$ $\gamma(t)=H(t,0)$ and $\mu(t)=H(t,1)$
\item For fixed $s_0$, $\eta_{s_0}:[0,1]\to X$ given by $\eta_{s_0}(t)=H(t,s_0)$ is  an execution (the coordinate functions are non decreasing).

\end{enumerate}
\end{samepage}
Conditions 1) and 2) ensure that $H$ is a \emph{homotopy} - condition 3) makes it a directed homotopy - it prescribes a continuous deformation of the execution path $\gamma$ to the execution path $\mu$ through execution paths. See Fig.~\ref{fig:homotopy}.  In e.g.  \cite{Papa79} and \cite{Pratt1}, this was mistakingly stated as  the executions paths being homotopic, i.e., $H(s_0,t)$ may not be execution paths. In \cite{LFEGMRAlgebraic} we give an example of a program and executions with homotopic but not directed homotopic execution paths which have different outcome with the same input. See Fig.~\ref{fig:localchoice}. This does not occur for directed homotopic execution paths.  For this to occur, there has to be at least three threads and resources cannot all have capacity $1$. In  Fig.~\ref{fig:localchoice} all execution paths  are homotopic but not all are directed homotopic.

\emph{Reparametrization} of an execution gives an equivalence: Let $\alpha: [0,1]\to [0,1]$ be surjective and non-decreasing and $\gamma:[0,1]\to X$ an execution path. The execution $\mu(t)=\gamma\circ \alpha(t)$ is equivalent to $\gamma$. Indeed $H(s,t)=\gamma((1-s)t+s\alpha(t))$ provides a directed homotopy. Hence, the image of $\gamma$ is the key information about an execution. 

\begin{example}\label{ex:elemhtpy}
Let $\gamma,\mu:[0,1]\to [0,1]\times [0,1]$ be execution paths of two threads with no shared resources or all resouces have capacity at least $2$ - all states are allowed. Then $\eta(t)=\max(\gamma(t),\mu(t))$ is an execution path and $H(s,t)=s\eta(t)+(1-s)\gamma(t)$ provides an equivalence of $\gamma$ and $\eta$. Similarly, $\mu$ is equivalent to $\eta$. All execution paths are equivalent as indeed they should be. In Fig.~\ref{fig:elemhtpy}, the equivalence of the two green paths can be constructed similarly. 
\end{example}

\emph{Subdivision and comparison:}  Executions which are a sequence of actions such as in Fig.~\ref{fig:serialize} a) the dotted path \[P^1c.V^1c.P^1a.V^1a.P^3a.V^3a.P^2c.V^2c.P^2b.V^2b.P^3b.V^3b\] 
 where $P^ir$ and $V^i r$ denote actions by thread $i$, and 
 \[T1.T2.T3=P^1c.V^1c.P^1a.V^1a.P^2c.V^2c.P^2b.V^2b.P^3a.V^3a.P^3b.V^3b\]  
are equivalent  if and only if there is a sequence of transpositions such as $ ....V^3a.P^2c....$ to $ .....P^2c.V^3a....$ which lead to this equivalence and such that the states in the rectangle between $(5,0,1)$ and $(5,1,2)$ are all allowed, see \cite{Fajstrup2005}. Then there is an \emph{elementary homotopy} as in  Ex.~\ref{ex:elemhtpy} and all partial execution paths from $(5,0,1)$ to $(5,1,2)$ are equivalent. In particular for subdivision of the actions such as e.g. $V^3a=u1.u2$ and $P^2c= w1.w2$, executions with $...u1.w1.u2.w2...$ are equivalent to other executions with this smaller stepsize; this also follows from \cite{Fajstrup2005}.  Moreover, we want to reason about  complicated parallel executions where perhaps in that example $P^3a.V^3 a$ executes while all actions from $T2$ execute -e.g. the dash-dotted path in  Fig.~\ref{fig:serialize} a). Such an execution path could come from  subdividing $P^3a.V^3 a$ into a total of $4$ actions $w1.w2.w3.w4$ and the concurrent execution of $w1|P^2c$,  $w2|V^2c$, $w3|P^2b$, $w4|V^2b$, each represented by a diagonal in a rectangle. This is why, we consider all continuous non-decreasing paths as executions. For a discussion about how this is a model of true concurrency which is robust to subdivision  and for comparison to other models, see e.g. \cite{Fajstrup2016} chapter 2.

\emph{Serializability:} An execution of $T1|T2|\cdots |Tn$ is serial if it executes one thread after the other - $Ti_1.Ti_2.\ldots .Ti_n$ where $i_1,\ldots,i_n$ is a permutation of $1,\ldots,n$. An execution of a parallel program is \emph{serializable}, if it is equivalent to a serial execution. A program is serializable if all executions of it are serializable. This is a strong requirement which stems from safety of databases. If all executions are equivalent to serial executions, verification is only needed for  serial executions where no conflicts occur wrt. shared resources. In database theory, it is well known that "two phase locking is safe" - if each $PV$-thread $Ti$ is two phased - lock all resources needed before starting to release any - then the program $p=T1|T2|\cdots |Tn$  is serializable (safe in databases is our serializable). But two phase locking is not necessary.

 In \cite{Papa79} an $O(s \log s \log\log s)$ algorithm to test for serializability in the case of two threads (called transactions there) with a total of $s$ steps (the number of $P$ and $V$),  is constructed based on the geometric interpretation, i.e.,  executions are paths  in the plane, non-decreasing in both coordinates and avoiding a union of  rectangles. With $n$ threads, the execution paths are in $\mathbb{R}^n$ and have to avoid a more complicated subset.

An \emph{obstruction to serializability} is a property of a state/point, such that if there are no obstructions at any state, then the program is serializable. We define such an obstruction,  a local choice point Def.~\ref{def:localchoicepoint}: a state, where a locally irreversible choice must be made. This is a refinement of \cite{MRMSCS}.

\begin{figure}
\centering
\begin{subfigure}{.45\textwidth}
\tdplotsetmaincoords{70}{110}
\begin{tikzpicture}[tdplot_main_coords, scale=0.5]
\draw[thick,->] (0,0,0) -- (5,0,0) node[anchor=north east]{$T1$};
\draw[thick,->] (0,0,0) -- (0,5,0) node[anchor=north west]{$T2$};
\draw[thick,->] (0,0,0) -- (0,0,5) node[anchor=south]{$T3$};
\draw[thick] (1,1,0)--(2,1,0)--(2,2,0)--(1,2,0)--(1,1,0);
\draw[thick] (1,1,5)--(2,1,5)--(2,2,5)--(1,2,5)--(1,1,5);
\draw[thick] (1,1,0)--(1,1,5);
\draw[thick] (2,1,0)--(2,1,5);
\draw[thick] (2,2,0)--(2,2,5);
\draw[thick] (1,2,0)--(1,2,5);
\foreach \i in {1.6,1.7,...,5}
\draw[blue] (2,1,\i)--(2,2,\i);
\foreach \i in {0.1,0.2,...,1.2}
\draw[blue] (2,1,\i)--(2,2,\i);
\draw[blue] (2,1,0.1)--(2,2,0.1);
\draw[thick] (3,0,1)--(4,0,1)--(4,0,2)--(3,0,2)--(3,0,1);
\draw[thick] (3,5,1)--(4,5,1)--(4,5,2)--(3,5,2)--(3,5,1);
\draw[thick] (3,0,1)--(3,5,1);
\draw[thick] (4,0,2)--(4,5,2);
\draw[thick] (3,0,2)--(3,5,2);
\draw[thick] (4,0,1)--(4,5,1);
\draw[thick] (0,3,3)--(0,4,3)--(0,4,4)--(0,3,4)--(0,3,3);
\draw[thick] (5,3,3)--(5,4,3)--(5,4,4)--(5,3,4)--(5,3,3);
\draw[thick] (0,3,3)--(5,3,3);
\draw[thick] (0,4,3)--(5,4,3);
\draw[thick] (0,4,4)--(5,4,4);
\draw[thick] (0,3,4)--(5,3,4);
\foreach \i in {0.05,0.1,...,5}
\draw[magenta] (3,\i,2)--(4,\i,2);
\draw[thick,dashed,green,rounded corners] (0,0,0)--(2,0,0)--(2,0,2)--(5,0,2)--(5,5,2)--(5,5,5);
\draw[thick,dotted,rounded corners] (0,0,0)--(5,0,0)--(5,0,1.95)--(5,5.05,1.95)--(5,5,5);
\draw[thick,dash dot,rounded corners] (5,0,0)--(5,5,2);
\end{tikzpicture}
\caption{Execution paths. The dashed green path $Pc^1.Vc^1.Pa^3.Va^3.Pa^1.Va^1.T2.Pb^3.Vb^3$ is non serializable.}
\end{subfigure}
\hspace{2mm}
\begin{subfigure}{.45\textwidth}
\tdplotsetmaincoords{70}{110}
\begin{tikzpicture}[tdplot_main_coords, scale=0.5]
\draw[thick,->] (0,0,0) -- (5,0,0) node[anchor=north east]{$T1$};
\draw[thick,->] (0,0,0) -- (0,5,0) node[anchor=north west]{$T2$};
\draw[thick,->] (0,0,0) -- (0,0,5) node[anchor=south]{$T3$};
\draw[thick] (1,1,0)--(2,1,0)--(2,2,0)--(1,2,0)--(1,1,0);
\draw[thick] (1,1,5)--(2,1,5)--(2,2,5)--(1,2,5)--(1,1,5);
\draw[thick] (1,1,0)--(1,1,5);
\draw[thick] (2,1,0)--(2,1,5);
\draw[thick] (2,2,0)--(2,2,5);
\draw[thick] (1,2,0)--(1,2,5);
\foreach \i in {1.1,1.2,...,2}
\draw[blue] (2,1,\i)--(2,2,\i);
\foreach \i in {1.1,1.2,...,2}
\draw[blue] (1,2,\i)--(2,2,\i);
\foreach \i in {1.1,1.2,...,2}
\draw[blue] (1,\i,2)--(2,\i,2);
\draw[thick] (1,0,1)--(2,0,1)--(2,0,2)--(1,0,2)--(1,0,1);
\draw[thick] (1,5,1)--(2,5,1)--(2,5,2)--(1,5,2)--(1,5,1);
\draw[thick] (1,0,1)--(1,5,1);
\draw[thick] (2,0,1)--(2,5,1);
\draw[thick] (2,0,2)--(2,5,2);
\draw[thick] (1,0,2)--(1,5,2);
\draw[thick] (0,1,1)--(0,2,1)--(0,2,2)--(0,1,2)--(0,1,1);
\draw[thick] (5,1,1)--(5,2,1)--(5,2,2)--(5,1,2)--(5,1,1);
\draw[thick] (0,1,1)--(5,1,1);
\draw[thick] (0,2,1)--(5,2,1);

\draw[thick] (0,2,2)--(5,2,2);

\draw[thick] (0,1,2)--(5,1,2);

\end{tikzpicture}
\caption{The program $Pa.Va|Pa.Va|Pa.Va$, $a$ has capacity $1$. }
\end{subfigure}
\caption{a) Execution paths in the parallel program $T1=Pc.Vc.Pa.Va$, $T2=Pc.Vc.Pb.Vb$, $T3=Pa.Va.Pb.Vb$, where all resources have capacity $1$. The blocks are states in which the resource use exceeds the capacity, i.e., where one of the resources is acceses by two threads. Resource $a$ in red, $c$ in blue and $b$ white. The blocks do not intersect. The dash-dotted path represents an execution of $P^3a.V^3a$ concurrently with as $P^2c.V^2c.P^2b.V^2b$.\newline b) The program $Pa.Va|Pa.Va|Pa.Va$, $a$ has capacity $1$. The $3$ boxes are states where a pair of the threads have a lock on $a$. The blue intersection is where they all have a lock on $a$. Here, all executions are serializable. }\label{fig:serialize}
\end{figure}

\begin{example}[Local choice points]\label{ex:localchoice}
$p=T1|T2|T3$ where $T1=Pd.Pc.Vc.Vd$, $T2=Pc.Pd.Vd.Vc$, $T3=Pc.Pd.Vd.Vc$ and  both resources have capacity $2$. The geometric representation of this  is seen in Fig.~\ref{fig:localchoice} b). The local choice is at the state where $T1$ requests $c$, which is locked by $T2$ and $T3$, $T2$ and $T3$ both request $d$ which is locked by $T1$. $T1$ cannot proceed, but \emph{either} $T_2$ or $T_3$ can proceed - locking $d$. This choice is only local - all executions are in fact serializable in this case. If however, before getting to this choice, other resources had been aquired - as in Fig.~\ref{fig:localchoice} a), $T1=Pa.Pd.Pb.Vb.Pc.Vc.Vd.Va$, $T2=Pa.Pb.Pc.Va.Pd.Vd.Vb.Vc$, $T3=Pa.Pb.Vb.Va.Pc.Pd.Vd.Vc$ all 4 resources have capacity $2$. Those may prevent undoing the local choice.  The execution indicated by the dotted path -  \[P^1a.P^1d.P^1b.V^1bP^1c.P^2a.P^2b.P^2c.V^2a.P^2d.T3.V^1c.V^1d.V^1a.V^2d.V^2b.V^2c\] is not serializable - see \cite{Raussen2006} for details - this geometric/topological result implies that all attempts at  pairwise permuting actions to get to a serial execution will fail - run into a permutation which is not allowed - for quite intricate reasons. 
\begin{figure}
\centering
\begin{subfigure}{.45\textwidth}
\tdplotsetmaincoords{70}{110}
\begin{tikzpicture}[tdplot_main_coords,scale=0.4]
\draw[thick,->] (9,0,0) -- (-2,-0.7,-0.7) node[anchor=south east]{$T2$};
\draw[thick,->] (9,0,0) -- (9,9,0) node[anchor=north west]{$T1$};
\draw[thick,->] (9,0,0) -- (9,0,9) node[anchor=south]{$T3$};

\draw[thick] (5,1,1)--(8,1,1)--(8,1,4)--(5,1,4)--(5,1,1);
\draw[thick] (5,8,1)--(8,8,1)--(8,8,4)--(5,8,4)--(5,8,1);

\draw[thick] (5,1,1)--(5,8,1);
\draw[thick] (8,1,1)--(8,8,1);
\draw[thick] (8,1,4)--(8,8,4);

\draw[thick] (5,1,4)--(5,8,4);

\draw[thick] (3,2,6)--(3,7,6)--(3,7,7)--(3,2,7)--(3,2,6);
\draw[thick] (4,2,6)--(4,7,6)--(4,7,7)--(4,2,7)--(4,2,6);
\draw[thick] (3,2,6)--(4,2,6);
\draw[thick] (3,7,6)--(4,7,6);
\draw[thick] (3,7,7)--(4,7,7);
\draw[thick] (3,2,7)--(4,2,7);

\draw[thick] (2,3,2)--(2,4,2)--(2,4,3)--(2,3,3)--(2,3,2);
\draw[thick] (7,3,2)--(7,4,2)--(7,4,3)--(7,3,3)--(7,3,2);
\draw[thick] (2,3,2)--(7,3,2);
\draw[thick] (2,4,2)--(7,4,2);
\draw[thick] (2,4,3)--(7,4,3);
\draw[thick] (2,3,3)--(7,3,3);
\draw[thick] (1,5,5)--(1,6,5)--(1,6,8)--(1,5,8)--(1,5,5);
\draw[thick] (6,5,5)--(6,6,5)--(6,6,8)--(6,5,8)--(6,5,5);
\draw[thick] (1,5,5)--(6,5,5);
\draw[thick] (1,6,5)--(6,6,5);
\draw[thick] (1,6,8)--(6,6,8);
\draw[thick] (1,5,8)--(6,5,8);
\draw[thick,dashed] (4,5,6)--(3,5,6)--(3,5,7)--(4,5,7)--(4,5,6);
\draw[thick,dotted] (5,3,2)--(5,4,2)--(5,4,3)--(5,3,3)--(5,3,2);
\path (4,5,6) node[circle, fill, inner sep=2]{};
\node[draw] at (4,8,3) {$a$};
\node[draw] at (1,0.5,6) {$d$};
\node[draw] at (4,4.5,8) {$c$};
\node[draw] at (8,3,3) {$b$};

\draw[thick,dotted,rounded corners] (9,0,0)--(8.9,5,0)--(4,5,0)--(4,5,9)--(4,9,9)--(0,9,9);

\end{tikzpicture}
\caption{The dotted execution path is not serializable.}
\end{subfigure}
\hspace{2mm}
\begin{subfigure}{.45\textwidth}
\tdplotsetmaincoords{70}{110}
\begin{tikzpicture}[tdplot_main_coords,scale=0.4]
\draw[thick,->] (9,0,0) -- (-2,-0.7,-0.7) node[anchor=south east]{$T2$};
\draw[thick,->] (9,0,0) -- (9,9,0) node[anchor=north west]{$T1$};
\draw[thick,->] (9,0,0) -- (9,0,9) node[anchor=south]{$T3$};

\draw[thick] (3,2,6)--(3,7,6)--(3,7,7)--(3,2,7)--(3,2,6);
\draw[thick] (4,2,6)--(4,7,6)--(4,7,7)--(4,2,7)--(4,2,6);
\draw[thick] (3,2,6)--(4,2,6);
\draw[thick] (3,7,6)--(4,7,6);
\draw[thick] (3,7,7)--(4,7,7);
\draw[thick] (3,2,7)--(4,2,7);

\draw[thick] (1,5,5)--(1,6,5)--(1,6,8)--(1,5,8)--(1,5,5);
\draw[thick] (6,5,5)--(6,6,5)--(6,6,8)--(6,5,8)--(6,5,5);
\draw[thick] (1,5,5)--(6,5,5);
\draw[thick] (1,6,5)--(6,6,5);
\draw[thick] (1,6,8)--(6,6,8);
\draw[thick] (1,5,8)--(6,5,8);
\draw[thick,dashed] (4,5,6)--(3,5,6)--(3,5,7)--(4,5,7)--(4,5,6);

\path (4,5,6) node[circle, fill, inner sep=2]{};

\node[draw] at (1,0.5,6) {$d$};
\node[draw] at (4,4.5,8) {$c$};

\draw[thick,dotted,rounded corners] (9,0,0)--(8.9,5,0)--(4,5,0)--(4,5,9)--(4,9,9)--(0,9,9);

\end{tikzpicture}
\caption{There is a choice point, but all execution paths are equivalent.}
\end{subfigure}

\caption{Local choice points. \newline a)  $T1=Pa.Pd.Pb.Vb.Pc.Vc.Vd.Va$, $T2=Pa.Pb.Pc.Va.Pd.Vd.Vb.Vc$, $T3=Pa.Pb.Vb.Va.Pc.Pd.Vd.Vc$. All 4 resources have capacity $2$ and give rise to 4 "boxes" where the capacity of a resource is exceeded. See Ex.~\ref{ex:localchoice}
There is a local choice point at the black dot where either $T2$ or $T3$ (but not both) proceeds locking $d$. See Ex.~\ref{ex:localchoice}.
 In this case, the obstruction is realized in the sense that there are non-serializable executions. Geometrically, the dotted execution path "gets caught" and cannot be deformed to a serial path through non decreasing paths.  This is the example "two wedges" - see \cite{Raussen2006} for the details. \newline b)  If  $T1=Pd.Pc.Vc.Vd$, $T2=Pc.Pd.Vd.Vc$, $T3=Pc.Pd.Vd.Vc$ - the resources $a$ and $b$ are not requested - the local choice is still there, but all executions are serializable. }\label{fig:localchoice} 

\end{figure}
\end{example}

\emph{Cut-offs for local choice points:} There are no local choice points in $T^n$ for any $n$ if and only if there are no local choice points in $T^M$ where $M=\Sigma_{r\in\mathcal{R}}\kappa(r)+1$.  A generalization of this is: If there are no local choice points in any of the  sub programs on $M$-threads, $T_{i_1}|T_{i_2}|\cdots |T_{i_M}$ of $p=T1|T2|\cdots |Tn$,  then $p$ is serializable.

Notice that this is an obstruction result, it is not an equivalence. A program with local choice points may be serializable as in Fig.~\ref{fig:localchoice} b).  Hence it is not a cut-off result for serializability: If there are no local choice points in $T^M$, then all $T^n$ are serializable. 

If there is a local choice point in $T^n$, then there are \emph{potential} deadlocks (deadlocks which may not be reachable from the initial point and may even be forbidden states, i.e., holding too many locks on a resource) in $T^{n+1}$, Prop.~\ref{prop:ifchoicethendead}. Hence, a deadlock algorithm may be used to rule out the existence of local choice points in $T^M$.

Thm.~\ref{thm:serializabilitykapacity2} states that for $\kappa\equiv 1$, i.e., when the shared resources are mutexes, $T^n$ is serializable for all $n$ if and only if $T^2$ is serializable. The general case, serializability of $T1|T2|\ldots Tn$ is NP-complete \cite{Papa79}. In \cite{Yannakakis1979} the requirement is specified as pairwise serializability plus a condition on a graph.

When all resources have capacity at least $2$, $T$ is serializable if and only if \emph{all} executions are equivalent. This is Thm.~\ref{thm:serializableiffconnected}.

\section{Preliminaries}

 The PV-programs considered in the present paper have no loops and no branchings/choice. Hence, the definitions here are not the most general ones. For the more general definition see \cite{Fajstrup2016}.  The geometric models do cover loops, even nested loops, and choice. Loops are studied via delooping, see \cite{LFMSCS}, where deadlocks are part of the study. A deadlock in a looped program is equivalent to a deadlock in a delooping, so our cut off for deadlocks hold in the case of loops. A deadlock in a program with choice is a deadlock for one of the choices, so our deadlock result holds for choice as well. Similarly, if executions are equivalent, they will have the same number of repeats of each loop, and also have the same choices,  so our serializability results apply to that case as well . We have chosen not to state the specifics of this in order not to complicate notation. 
 First some notation:
\begin{itemize}
\item $\N$ are the natural numbers. $\N_0$ are natural numbers and $0$. $\R$ are the real numbers.
\item For $a<b$ real numbers, $[a,b]$ is the closed interval, $]a,b[$ is the open interval and half open intervals are $]a,b]$, $[a,b[$.
\item $I=[0,1]$ the unit interval.
\item For $n,m\in\N_0$ , $[n:m]$ is the set of  $k\in\N_0$ s.t $n\leq k\leq m$.
\item The coordinate functions of $\gamma: I\to \R^n$ are called $\gamma_i$.  $\gamma(t)=(\gamma_1(t),\gamma_2(t),\ldots,\gamma_n(t))$.

\end{itemize}
 \begin{definition}\label{def:valid} 
 Given a set $\mathcal{R}$ of resources, each with a positive capacity $\kappa:\mathcal{R}\to \N$. 
A \emph{PV-thread} is a finite sequence $T=w_1w_2\ldots w_l$ where $w_i\in\{Pr,Vr|r\in\mathcal{R}\}$. 

The resource use for a fixed  $r\in\mathcal{R}$ is defined for $0\leq i\leq l+1$:
\begin{itemize}
\item $\rho_r(T,0)=0$
\item For $i>0$: 
\begin{itemize}
\item $\rho_r(T,i)=\rho_r(T,i-1)+1$ if $w_i=Pr$.
\item $\rho_r(T,i)=\rho_r(T,i-1)-1$ if $w_i=Vr$.
\item $\rho_r(T,i)=\rho_r(T,i-1)$ otherwise.
\end{itemize}
\end{itemize}

$T$ is \emph{valid} if $0\leq\rho_r(T,i)\leq 1$ for all $i$ and $\rho_r(l+1)=0$ for all $r\in\mathcal{R}$.

A (valid) \emph{PV-program} is a parallel composition of (valid) PV-threads $p=T1|T2|\cdots|Tn$. If $T1=T2=\cdots =Tn=T$, this is denoted $p=T^n$.
 A state of $p$ is a tuple $\mathbf{x}=(x_1,\ldots,x_n)$, where $x_j\in [0,l_j+1]$.
The initial state, $\mathbf{0}$ is denoted $\perp$ and the final state where $x_j= [0,l_j+1]$  is $\top$.
The resource use at  $\mathbf{x}=(x_1,\ldots,x_n)$ is $\rho_r(p,\mathbf{x})=\Sigma_{j=1}^n\rho_r(T_j,x_j)$.
\end{definition}
\begin{remark} If it is clear what thread or program is considered, $\rho_r(-,j)$ will be denoted $\rho_r(j)$. 
\end{remark}
\begin{lemma}\label{lem:PVsequence} Let $T=w_1w_2\cdots w_l$ be a valid PV-thread. 
For each $r\in\mathcal{R}$ let $\mathcal{P}(r)=\{i\mid w_i=Pr\}$ and $\mathcal{V}(r)=\{i\mid w_i=Vr\}$. When these sets are non empty, let $i_1<i_2<\ldots <i_{k(r)}$ and $j_1<j_2<\ldots <j_{m(r)}$ be the ordered elements of $\mathcal{P}(r)$ respectively $\mathcal{V}(r)$
Then 
\begin{itemize}
\item For all $r$:  $\#\mathcal{P}(r)=\#\mathcal{V}(r)$, i.e., $k(r)=m(r)$
\item $i_s<j_s<i_{s+1}$ for $s=1,\ldots, k(r)-1$

\end{itemize}

\end{lemma}
\begin{proof}
First, by Def.\ref{def:valid}, 
$\rho_r(i)=\#\{s\in\mathcal{P}(r)\mid s\leq i\}-\#\{s\in\mathcal{V}(r)\mid s\leq i\}$ and $\rho_r(l+1)=0$.
Hence, either both sets are empty or
$k(r)=m(r)$.

$i_k<j_k$: If $i_k>j_k$, then $\rho_r(j_k)\leq -1$ which is not allowed for a valid thread.

$j_k<i_{k+1}$: If $j_k>i_{k+1}$, then $\rho_r(i_{k+1})\geq 2$ and again this violates validity of $T$.

Suppose $w_1= Va$ for some $a\in\mathcal{R}$, then the $Pa,Va$-sequence does not satisfy the above results. Hence, $w_1=P\hat{r}$ for some $\hat{r}\in \mathcal{R}$. Similarly, $w_l$ has to be release of a resource. 
\end{proof}
Only valid PV-threads and valid PV-programs are considered in the following. 

\begin{definition}\label{def:PrVrsequence} With notation from Lem.~\ref{lem:PVsequence}   The sequence $i_1<j_1<i_2<\ldots<i_{k(r)}<j_{k(r)}$ is the $Pr,Vr$ sequence for $T$.
In particular, $w_1=P\hat{r}$ for some $\hat{r}\in \mathcal{R}$ and $w_l=V{\tilde{r}}$ for some $\tilde{r}\in\mathcal{R}$.
\end{definition}
In \cite{FGHMR2012}, see also \cite{Fajstrup2016} p.62, a geometric model of a more  general PV-program is provided. For our simpler case, it is as follows:
\begin{definition}\label{def:geommodel}
 The geometric model of a valid thread $T=w_1w_2\ldots w_l$ is the interval $[0,l+1]$.
 Let $i_1<j_1<i_2<\ldots<i_{k(r)}<j_{k(r)}$ be the $Pr,Vr$ sequence for $T$. The resource use $\rho_r:[0,l+1]\to \{0,1\}$ is given by $\rho_r(t)=1$ for $t\in ]i_k,j_k[$, $k=1,\ldots, k(r)$ and $\rho_r(t)=0$ otherwise.
 
 The geometric model of a PV-program $p=T1|T2|\cdots|Tn$ is the subset of the $n$-rectangle $$X=\{(x_1,\ldots,x_n)\in [0,l_1+1]\times \cdots \times [0,l_n+1]\mid \forall r\in\mathcal{R}:\Sigma_{i=1}^n\rho_r(x_i)\leq\kappa(r)\}$$
 
 The point $(0,0,\ldots,0)$ is denoted $\perp$ and the top point $(l_1+1,\ldots,l_n+1)$ is denoted $\top$. Notice the slight abuse of notation: The coordinate $l_i+1$ in the thread $Ti$ is also denoted $\top$. If clarification is needed, this is called $\top_i$ and similarly $\perp_i$ denotes $x_i=0$.
 
 Points in $X$ are  states of the program. A coordinate $x_i$ is a state of $Ti$.  An integer coordinate $x_i$ corresponds to either access request $Pr$, release, $Vr$, bottom $\perp_i$ or top $\top_i$ of $Ti$. 
\end{definition}
\begin{remark} The definition of the resource use function for a thread has the following interpretation: The thread $T$ holds a lock on the  resource in the open interval $]i_k,j_k[$, i.e., it is requested, but not yet granted at $x=i_k$. It is released and not held anymore at $x=j_k$. This corresponds to considering $T$ as a linear graph, where $Pa$ is a state which is followed by the action - an edge -  of actually locking $a$ and $Va$ is the state right after an edge releasing $a$. In \cite{Fajstrup2016} p. 62, the locking is shifted by $\frac{1}{2}$ in the sense that $Pa$ and $Va$ both correspond to a unit interval, an edge, so that if $Pa=[2,3]$ and $Va=[3,4]$, the resource is locked in the interval $]2.5, 3.5[$. Our choice here ensures that e.g. a deadlock point has integer coordinates. 
\end{remark}
The state space $X$ consists of points in $[0,l_1+1]\times \cdots \times [0,l_n+1]$ where no resource is locked above its capacity.
\begin{example}\sloppy In Fig.~\ref{fig:serialize}b)  the program $Pa.Va|Pa.Va|Pa.Va$,  $\kappa(a)=1$, gives rise to three forbidden 3-rectangles. $]1,2[\times ]1,2[\times [0,3]$,   $]1,2[\times [0,3]\times ]1,2[$, $[0,3]\times ]1,2[\times ]1,2[$. These are the states where at least $2$ threads lock the resource. If $\kappa(a)=2$  the forbidden region is $]1,2[\times ]1,2[\times]1,2[$, the states where all $3$ threads hold a lock on $a$. If $a$ had capacity $3$, all states would be allowed.
\end{example}

The following definition and lemma describe this  complement of $X$, the \emph{forbidden states} as a union of $n$-rectangles. 
\begin{definition}\label{def:forbidden}
The forbidden states are states where some resource is locked above its capacity. The \emph{forbidden area} is the set of all such states:  $F_p=\{(x_1,\ldots,x_n)\in [0,l_1+1]\times \cdots \times [0,l_n+1]\mid\Sigma_{i=1}^n\rho_r(x_i)>\kappa(r)\;\mbox{for some}\; r\in\mathcal{R}\}$. 

The geometric model is  $X=[0,l_1+1]\times \cdots \times [0,l_n+1]\setminus F_p$
\end{definition}

\begin{lemma}\label{lem:geommodel} With notation as in Def.~\ref{def:geommodel} and Def.~\ref{def:forbidden} the forbidden area is the union $F_p=\bigcup_{m=1}^N R^m$ of all $n$-rectangles  $R^m=I_1\times\cdots\times I_n$ for which there is a resouce $r$ s.t.
\begin{itemize}
\item There is a subset $\mathcal{I}=\{ k_1,\ldots , k_{\kappa(r)+1} \}\subseteq [1: n]$ s.t. $I_{k_i}= ]i_s,j_s[$, where $i_s,j_s$ are  in the $Pr,Vr$-sequence for $Tk_i$.
\item For $j\notin \mathcal{I}$, $I_j=[0,l_j+1]$.
\end{itemize}

\end{lemma}
\begin{proof}
$F_p\subseteq\bigcup_{m=1}^N R^m$: 
Let $(x_1,\ldots, x_n)\in F_p$ and let $r\in\mathcal{R}$ be a resource s.t. $\Sigma_{i=1}^n\rho_r(x_i)>\kappa(r)$. 

As $\rho_r(x_i)\in \{0,1\}$, $\Sigma_{i=1}^n\rho_r(x_i)>\kappa(r)$ if and only if there is a set $\mathcal{I}=\{ k_1,\ldots , k_{\kappa(r)+1} \}\subset[1: n]$, s.t.  $\rho_r(x_{k_i})=1$. 

This implies that $x_{k_i}\in ]i_s,j_s[$, where $i_s,j_s$ are  in the $Pr,Vr$-sequence for $Tk_i$. For $j\notin \mathcal{I}$ there is no restriction and hence $(x_1,\ldots,x_n)\in I_1\times\cdots\times I_n$ with $I_{k_i}=]i_s,j_s[$ for $i=1,\ldots,\kappa(r)+1$ and $I_j=[0,l_j+1]$ otherwise. 

$F_p \supseteq\bigcup_{m=1}^N R^m$: Let $R$ be an $n$-rectangle as above with corresponding resource $r$. If $\mathbf{y}\in R$, then  $\rho_r(\mathbf{y})>\kappa(r)$ and hence $\mathbf{y}\in F_p$.

\end{proof}
The following is a quite general definition which we give to define execution paths and to extend partial executions
\begin{definition}\label{def:dipath}
Let $X$ be a subset of $\mathbb{R}^n$. A continuous function  $\gamma:[0,1]\to X$, $\gamma(s)=(\gamma_1(s),\gamma_2(s),\ldots,\gamma_n(s))$ such that $t\leq s \Rightarrow \gamma_i(t)\leq \gamma_i(s)$ for $i\in [1:n]$ is called a \emph{dipath}.

Let $\gamma_1$ and $\gamma_2$ be dipaths in $X$ s.t. $\gamma_1(1)=\gamma_2(0)$ then the \emph{concatenation}  is the dipath $\gamma_1\star\gamma_2: [0,1]\to X$ with $\gamma_1\star\gamma_2(t)= \gamma_1(2t)$ for $t\in [0,\frac{1}{2}]$ and $\gamma_1\star\gamma_2(t)= \gamma_2(2t-1)$ for $t\in [\frac{1}{2},1]$ 
\end{definition}

\begin{definition}\label{def:execution}

A \emph{partial execution} of a PV-program with geometric model $X$  is a dipath  $\gamma:[0,1]\to X$,  such that $\gamma(0)=\perp$.

An execution is a partial execution such that $\gamma(1)=\top$

A state $\mathbf{y}\in X$ is \emph{reachable} if there is a partial execution $\gamma$ with $\gamma(1)=\mathbf{y}$. 

 Executions  $\gamma$ and $\mu$ are equivalent  if there is  a continuous map $H:I\times I\to X$ such that: For all $s$ $H(0,s)=\perp$, $H(1,s)=\top$,  for fixed $s_0$, $H(t,s_0)$ is an execution path, i.e., $t\leq t'\Rightarrow H(t,s_0)\leq H(t',s_0)$.  Moreover, $H(t,0)=\gamma(t)$, $H(t,1)=\mu(t)$. Such an $H$ is a \emph{dihomotopy} and the execution paths $\gamma$ and $\mu$ are \emph{dihomotopic}. 

\end{definition}
For examples of executions, equivalent executions and reasons for this definition, please see the introduction.

\begin{figure}
\centering
\begin{tikzpicture}[scale=0.6]
 
  \node[label=below:$\perp$]  (x1) at (6,0)  {$\bullet$};
  \node[label=above:$\top$]  (x0) at (9,4)  {$\bullet$};  
  \draw[green] (x1.center) to [out=5,in=-90]++(2.8,1.8) to[out=90,in=-95](x0.center);
  \draw (x1.center) to [out=10,in=-110]++(2.6,2) to[out=70,in=-103](x0.center); 
  \draw (x1.center) to [out=15,in=-105](x0.center);
  \draw (x1.center) to [out=30,in=-150](x0.center);
  \draw (x1.center) to [out=45,in=-170](x0.center); 
  \draw (x1.center) to [out=50,in=-105]++(1.2,3)to [out=75,in=-172](x0.center); 
  \draw (x1.center) to [out=55,in=-100]++(1.0,3) to[out=80,in=-175](x0.center); 
  \draw[red] (x1.center) to [out=60,in=-90]++(0.8,3) to[out=90,in=-180] (x0.center);
  \begin{scope}[every node/.style={draw, anchor=text, rectangle split,
    rectangle split parts=7,minimum width=2cm}]
    \node (R) at (2,4){ \nodepart{two} \nodepart{three}
    \nodepart{four}$I\times I$\nodepart{five}\nodepart{six}\nodepart{seven}};
  \end{scope}
\draw[green](R.south east)--(R.south west);
\draw[red](R.north east)--(R.north west);
  \draw[decorate,decoration={brace,mirror,raise=6pt,amplitude=10pt}, thick]
    (R.north west)--(R.south west) ;
  \draw[decorate,decoration={brace,raise=6pt,amplitude=10pt}, thick]
    (R.north east)--(R.south east); 
  \draw[dotted,->] ($(R.west)+(-20pt,0)$) to[out=180,in=-240] ++(0,-2)
    to [out=-60,in=-120]node[below,midway]{$H(0,s)$}(x1) ; 
  \draw[->] ($(R.north)+(0,10pt)$) to [out=60,in=120]
    node[above,midway]{$\mu$} ++(4.5,-1) ; 
  \draw[dotted, ->] ($(R.east)+(20pt,0)$)  to [out=0,in=140]
    node[left,midway]{$H(1,s)$}(x0) ; 
  \draw[->] ($(R.south)+(0,-20pt)$)  to [out=-85,in=-30]
    node[below,midway]{$\gamma$}++(7,0) ;    
\end{tikzpicture}
\caption{A directed homotopy - a continuous family of execution paths.}\label{fig:homotopy}
\end{figure}

\section{Deadlock}
Deadlocks in a $PV$-program are characterized in terms of resource use and capacities and the equivalent geometric definition is given. The main result is the cut-off theorem for deadlocks in $T^n$, Thm.~\ref{thm:deadgeneral} and Cor.~\ref{cor:deadsymmetric}. By \cite{LFMSCS}, if there are loops in $T$, a state is a deadlock in $T^n$ if and only if the corresponding state is a deadlock in the non-looped program, where all loops in $T^n$ are delooped once in the sense that all $w_i( w_{i+1}...w_k)^*w_{k+1}$ are replaced by $w_i. w_{i+1}...w_k.w_{k+1}$ . Similarly for non deterministic choice: A deadlock in such a program is a deadlock for at least one of the choices. So deadlocks can be found one choice at a time.  Hence, the deadlock cut-off results hold for more general $PV$-programs.

\begin{definition}\label{def:deadlock}

Let $X$ be the geometric model of a $PV$-program with $n$ threads. 
The point $\mathbf{x}=(x_1,\ldots,x_n)\in X$ is a geometric deadlock if all of the following hold
\begin{enumerate}
\item $\mathbf{x}$ is reachable and not $\top$.
\item If $\gamma:I\to X$ is a partial execution with $\gamma(t_0)=\mathbf{x}$ then $\gamma(t)=\mathbf{x}$ for all $t\in [t_0,1]$. Execution paths cannot proceed from $\mathbf{x}$.
\end{enumerate}
Equivalently - as seen in Lem.~\ref{lem:deadlock}
A state $\mathbf{x}=(x_1,\ldots,x_n)\in X$ is a deadlock if the following three conditions hold
\begin{enumerate}
\item $\mathbf{x}$ is reachable and not $\top$.
\item All $x_i$ are access requests $P_{r(i)}$ or $x_i=\top$.
\item For all  $i\in [1:n]$, s.t. $x_i\neq \top$, $\rho_{r(i)}(x_1,\ldots,x_n)= \kappa(r(i))$.
\end{enumerate}

\end{definition}
The geometric definition is as in \cite{LFEGMRAlgebraic} Def.~4.44. The following lemma ensures that our definition of the forbidden area and executions does not allow a complicated execution path to escape from $\mathbf{x}$ by e.g. execution of several threads concurrently, as long as none of the individual threads can proceed. 

\begin{lemma}\label{lem:deadlock}
The two definitions of a deadlock are equivalent. In particular all coordinates of a geometric deadlock $\mathbf{x}$ are natural numbers.
\end{lemma}
\begin{proof} 
If, contrary to 3.1.1, a partial  execution path can proceed from $\mathbf{x}$, then at least one of the threads can proceed. This is because the forbidden region is a union of $n$-rectangles. With general geometric shapes, it would not hold. We make this precise in \cite{LFEGMRAlgebraic} Thm. 5.11 where in the proof of i) implies ii) we see that a state  $\mathbf{x}$ is a geometric deadlock if and only if it is reachable, $\mathbf{x}\neq \top$ and for all $i$, either $x_i=\top$ or  $(x_1,x_2,\ldots ,x_{i-1},x_i+t,x_{i+1},\ldots,x_n)\in F_p$ for  all $t\in ]0,\frac{1}{2}[$. I.e., no execution path can proceed from $\mathbf{x}$ if and only if no individual thread can proceed.

Suppose $\mathbf{x}$ is a deadlock.
If $x_j=\top$, the j'th process cannot proceed. If $x_i=P_{r(i)}$ and $\rho_{r(i)}(\mathbf{x})=\kappa(r(i))$, then $\rho_{r(i)}(x_1,x_2,\ldots ,x_{i-1},x_i+t,x_{i+1},\ldots,x_n)=\kappa(r(i))+1$, for $0<t<\frac{1}{2}$ and therefore $(x_1,x_2,\ldots ,x_{i-1},x_i+t,x_{i+1},\ldots,x_n)\in F_p$ So $\mathbf{x}$ is a geometric deadlock.

Now suppose $\mathbf{x}$ is a geometric deadlock. If $x_i=Vr$, then the i'th process can proceed and $\mathbf{x}$ is not a geometric deadlock, hence $x_i=P_{r(i)}$ or $x_i=\top$. If $x_i=P_{r(i)}$ and $\rho_{r(i)}<\kappa(r(i))$, then the i'th process may proceed and $\mathbf{x}$ is not a deadlock. Hence, $\rho_{r(i)}(\mathbf{x})\geq \kappa(r(i))$, and $\rho_{r(i)}(\mathbf{x})\leq \kappa(r(i))$ as $\mathbf{x}\in X$, so $\rho_{r(i)}(\mathbf{x})= \kappa(r(i))$.

\end{proof}

\begin{lemma}\label{lem:topdeadlock}
\sloppy Let $T1,T2,\ldots,Tn$ be PV-threads. If $(x_1,\ldots,x_m)$ is a deadlock in $Ti_1|\ldots|Ti_m$ where $i_j\in\{1,\ldots,n\}$  then $(x_1,\ldots,x_m,\top,\ldots,\top)$ is a deadlock in $Ti_1|\ldots|Ti_m|T{j_1}|\ldots|T{j_k}$ for all choices of $j_i\in\{1,\ldots,n\}$
\end{lemma}
\begin{proof}
 \sloppy The resource use $\rho_s(x_1,\ldots,x_m)=\rho_s(x_1,\ldots,x_m,\top,\ldots,\top)$ for all resources $s$, as no resources are held at $\top$. The requests for resources are also the same, namely $x_i=Pr(i)$ or $x_i=\top$. Hence, $ (x_1,\ldots,x_m,\top,\ldots,\top)$ satisfies point 2) and 3) of Def.~\ref{def:deadlock}. This point is reachable by a concatenation of
\begin{itemize}
\item A sequential execution path $Tj_1.Tj_2.\ldots .Tj_k$ to $(0,\ldots,0,\top,\ldots,\top)$ (the first $m$ coordinates stay $0$.)
\item  Followed by a dipath from $(0,\ldots,0,\top,\ldots,\top)$ to $(x_1,\ldots,x_m,\top,\ldots \top)$ given as follows: $(x_1,\ldots,x_m)$ is reachable by a dipath $\gamma$ in  $Ti_1|\ldots|Ti_m$ from $\mathbf{0}$. Let $\mu=(\gamma_1(t),\ldots,\gamma_m(t),\top,\ldots,\top)$ in $\R^n$, $\mu(0)=(0,\ldots,0,\top,\ldots,\top)$ and $\mu(1)=(x_1,\ldots,x_m,\top,\ldots \top)$. $\mu(t)$ is in $X$, as no resources are locked at $\top$ and hence $\rho_r(\mu(t))=\rho_r(\gamma(t))$ for all $r$.
\end{itemize}
\end{proof}

\begin{example}\label{ex:deadcombo} Let $T1=Pa.Pb.Vb.Va$ and $T2=Pb.Pa.Va.Vb$ and let both resources have capacity 1. Then $T1|T2$ has a deadlock. The thread $T=T1.T2$ in parallel with itself $T^2$ has two deadlocks. See Fig.~\ref{fig:twodeadlocks}. This generalizes, see Prop.~\ref{prop:deadcombo}.  If $T1|T2|\cdots |Tn$ has a deadlock, then $(T1.T2.\ldots.Tn)^n$ has at least $n(n-1)$ deadlocks.
\end{example}

\begin{figure}
\begin{center}
\begin{tikzpicture}[scale=0.6]
\draw [->] (0,0)--(5,0);
\draw [->] (0,0)--(0,5);

\filldraw[fill=white!30!blue,opacity=.8,draw=black] (1,2) rectangle (4,3);
\filldraw[fill=white!30!blue,opacity=.8,draw=black] (2,1) rectangle (3,4);

\filldraw[fill=red] (2,2) circle (.08cm);
\draw (1,0) node[anchor=north]{$Pa$};
\draw (2,0) node[anchor=north]{$Pb$};
\draw (3,0) node[anchor=north]{$Vb$};
\draw (4,0) node[anchor=north]{$Va$};

\draw (0,1) node[anchor=east]{$Pb$};
\draw (0,2) node[anchor=east]{$Pa$};
\draw (0,3) node[anchor=east]{$Va$};
\draw (0,4) node[anchor=east]{$Vb$};

\draw (5,0) node[anchor=west]{$T1$};
\draw (0,5) node[anchor=south]{$T2$};
\foreach \x in {1,...,4}
\draw (\x cm,1pt)--(\x cm,-1pt);
\foreach \y in {1,...,4}
\draw (-1pt,\y cm)--(1pt,\y cm);

\end{tikzpicture}
\begin{tikzpicture}[scale=0.6]
\draw [->] (0,0)--(9,0);
\draw [->] (0,0)--(0,9);

\filldraw[fill=white!20!blue,opacity=.8] (2,2) rectangle (3,3);
\filldraw[fill=white!30!blue,opacity=.8,opacity=.8,draw=black] (1,1) rectangle (4,4);
\filldraw[fill=white!20!blue,opacity=.8] (5,5) rectangle (8,8);
\filldraw[fill=white!30!blue,opacity=.8,opacity=.8,draw=black] (6,6) rectangle (7,7);
\filldraw[fill=white!30!blue,opacity=.8,draw=black] (6,1) rectangle (7,4);
\filldraw[fill=white!30!blue,opacity=.8,draw=black] (1,6) rectangle (4,7);
\filldraw[fill=white!30!blue,opacity=.8,draw=black] (2,5) rectangle (3,8);
\filldraw[fill=white!40!blue,opacity=.8,draw=black] (5,2) rectangle (8,3);
\filldraw[fill=red] (6,2) circle (.08cm);
\filldraw[fill=red] (2,6) circle (.08cm);
\draw (1,0) node[anchor=north]{$Pa$};
\draw (2,0) node[anchor=north]{$Pb$};
\draw (3,0) node[anchor=north]{$Vb$};
\draw (4,0) node[anchor=north]{$Va$};
\draw (5,0) node[anchor=north]{$Pb$};
\draw (6,0) node[anchor=north]{$Pa$};
\draw (7,0) node[anchor=north]{$Vb$};
\draw (8,0) node[anchor=north]{$Va$};

\draw (0,1) node[anchor=east]{$Pa$};
\draw (0,2) node[anchor=east]{$Pb$};
\draw (0,3) node[anchor=east]{$Vb$};
\draw (0,4) node[anchor=east]{$Va$};
\draw (0,5) node[anchor=east]{$Pb$};
\draw (0,6) node[anchor=east]{$Pa$};
\draw (0,7) node[anchor=east]{$Va$};
\draw (0,8) node[anchor=east]{$Vb$};

\draw (9,0) node[anchor=west]{$T$};
\draw (0,9) node[anchor=south]{$T$};
\foreach \x in {1,...,8}
\draw (\x cm,1pt)--(\x cm,-1pt);
\foreach \y in {1,...,8}
\draw (-1pt,\y cm)--(1pt,\y cm);

\end{tikzpicture}
\end{center}
\caption{One deadlock in $T1|T2$. Two deadlocks in $(T1.T2)^2$.}\label{fig:twodeadlocks}

\end{figure}

\begin{example} Let $T$ be  $PV$-thread such that every resource is accessed at most once, then there are no deadlocks in $T^n$ and no need for a cut-off result. Suppose $\mathbf{x}=(x_1,\ldots,x_n)$ is a deadlock. Then, if $x_i\neq\top$, $x_i=P{r(i)}$. Since $\mathbf{x}$ is a deadlock, there are $\kappa(r(i))$ locks on $r(i)$. The threads $j$ which hold a lock on $r(i)$ satisfy $x_i<x_j<\top$, since $r(i)$ is only locked once - at $x_i$. 

Let $x_k=\max(\{x_1,\ldots,x_n\}\setminus\top)$. Then $x_k=P{r(k)}$ and  $r(k)$ is not locked, since $x_k$ is maximal. A contradiction.
\end{example}

The symmetric case $T^n$ is as complicated as different threads in parallel in the following sense:

\begin{proposition}\label{prop:deadcombo} Let $p=T1|T2|\ldots|Tn$ be a valid PV-program. Suppose there is a  deadlock in $p$. Let $T=T1.T2.\ldots.Tn$, then there are at least $n(n-1)$  deadlocks in $T^n$. If all threads are nontrivial, there are at least $n!$ deadlocks in $T^n$.

\end{proposition}
\begin{proof} \sloppy Let $\mathbf{x}=(x_1,\ldots,x_n)$ be  a  deadlock in $p$. Let $l_i$ be as in Def.~\ref{def:geommodel} . The point $\tilde{\mathbf{x}}=(x_1, l_1+x_2, l_1+l_2+x_3,\ldots,\Sigma_{i=1}^{n-1} l_i+x_n)$ is a deadlock in $T^n$:  The action and locked resources at  $\Sigma_{i=1}^j +x_{j+1}$ in $T$ is the same as at $x_{j+1}$ in $T(j+1)$. $\tilde{\mathbf{x}}$ is reachable by a concatenation $\eta\star\mu$ of a sequential dipath $\eta$ to $ \mathbf{y}=(0, l_1, l_1+l_2,\ldots, \Sigma_{i=1}^{n-1} l_i)$ and $\mu(t)=\mathbf{y}+\gamma(t)$ - addition of vectors, where $\gamma$ is a dipath from $\perp$ to $\mathbf{x}$ in $p$. Hence, $\tilde{\mathbf{x}}$ is a deadlock. By symmetry, all permutations of the coordinates in $\tilde{\mathbf{x}}$ give deadlocks. As  $x_i\neq 0$ for all $i$ and as there is a deadlock in $p$, at least $2$ of the threads are non trivial, there are at least $2$ different coordinates in $\tilde{\mathbf{x}}$ and hence at least  $n(n-1)$ such deadlocks. If all threads are non trivial, $l_i> 0$, all coordinates of $\tilde{x}$ are different - all $x_i\neq 0$ - and the permutations give rise to $n!$ different deadlocks.

\end{proof}
\begin{remark} Proposition \ref{prop:deadcombo} shows how construct a program $T^n$ from a $PV$-program $p=T1|T2|\ldots|Tn$ in such a way that the states and executions of $p$ have counterparts in $T^n$ with the same properties. This hints at ways of using our results in the non-symmetric setting: Suppose the program $p$ accesses resources of total capacity $M$. Then there are deadlocks in  $T^k$ if and only if there are deadlocks in $T^M$. If there are no deadlocks in $T^M$, then there are no deadlocks in $T^n$ and hence, by the proposition, no deadlocks in $p$. The usefulness of this will of course depend on $M$.
\end{remark}
\begin{remark}\label{rem:compositionality}[Compositionality?] The converse of the Prop.~\ref{prop:deadcombo} does not hold as illustrated by  $T1=Pa.Pb.Vb.Va.Pb.Pa.Va.Vb$ and $T2=Pc.Vc$, all resources have capacity $1$. There are no deadlocks in $T1|T2$ but in $(T1.T2)^2$ there are deadlocks at $(2,6)$ and at $(6,2)$, where one holds a lock on $a$ and request $b$ and the other holds a lock on $b$ and requests $a$.

As a consequence, even if $T$ decomposes in valid threads $T=T1.T2.\ldots.Tk$, the deadlock analysis of $T^n$ is not equivalent to analyzing  $T1|T2|T3\ldots |Tk$. However, as coordinates of a deadlock $\mathbf{x}$ in $T^n$ will all be in on of the $Ti$,  by ordering the coordinates, all deadlocks will be found in one of $T1^{m_1}|T2^{m_2}\ldots|Tk^{m_k}$ where $m_i \in \mathbb{N}_0$, $\Sigma_{i=1}^km_i=n$ and $Tj^0$ indicates skip $Tk$ in the sense that $T1^n|T2^0$ is  $T1^n$. We will not go deeper into that here. It is related to compositionality as in \cite{NamjoshiTrefler12}, \cite{NamjoshiTrefler16} in the sense that ordering the coordinates is a choice of a representative of the orbit of $\mathbf{x}$ under the action of the symmetric group (permutation of coordinates).
\end{remark}

\begin{theorem}\label{thm:deadgeneral} Let $p=T1|T2|\ldots|Tn$ be a valid PV- program accessing resources in $\mathcal{R}$. There is a deadlock in $p$ if and only if there is a deadlock in a PV-program $Ti_1|\ldots|Ti_m$ for some subset $i_1<i_2<\ldots i_m\in [1:n]$ of at most $M$ threads, where $M={\Sigma_{r\in\mathcal{R}}\kappa(r)}$.
\end{theorem}
\begin{proof} 
If there is a deadlock $(x_1,\ldots,x_m)$ in $Ti_1|\ldots|Ti_m$ , then by reordering the threads in $p$ this gives a deadlock at $(x_1,\ldots,x_m,\top,\ldots,\top)$ in p, by Lem.~\ref{lem:topdeadlock}.

Suppose $\mathbf{x}=(x_1,\ldots,x_n)$ is a deadlock for p. Then all threads either are at $\top$ or request a resource $r$ which is held by $\kappa(r)$ other threads. Let $\tilde{x}=(x_{i_1},\ldots,x_{i_m})$ be the coordinates for which $Ti_j$ both holds a resource, i.e., $\rho_l(x_{i_j})>0$ for some $l\in\mathcal{R}$, and requests a(nother) resource. This is still a deadlock, since 
\begin{itemize}
\item The resource use is the same: If $Tj$ holds a resource at $x_j$, then $x_j\neq\top$, since no resources are held at $\top$. Hence $x_j=P{r(j)}$, so $Tj$ both holds and requests a resource and thus $x_j=x_{i_s}$ for some $s$ .
\item All $x_{i_j}=P{r(i_j)}$ and $\rho_{r(i_j)}(\tilde{x})=\rho_{r(i_j)}(\mathbf{x})=\kappa(r(i_j))$.
\item $\tilde{x}$ is reachable: Let $\gamma$ be a partial execution with $\gamma(1)=\mathbf{x}$. Then the restriction $\mu(t)= (\gamma_{i_1}(t),\ldots, \gamma_{i_m}(t))$ defines a partial execution of $Ti_1|\ldots|Ti_m$ with $\mu(1)=\tilde{x}$ .
\end{itemize}
At most $M$ processes can hold a resource, so $m\leq M$. 

\end{proof}
Our cut-off for deadlocks is a corollary:
\begin{corollary}\label{cor:deadsymmetric} Let $T$ be a valid $PV$ thread. Then $T^n$ is deadlock free for all $n$ if and only if $T^M$ is deadlock free, where $M={\Sigma_{r=1}^k\kappa(r)}$.
\end{corollary}
\begin{proof} If $T^n$ is deadlock free for all $n$, clearly  $T^M$ is deadlock free. 

Suppose there is a deadlock in $T^n$ for some $n$. Then there is a deadlock in $T^M$:

\begin{enumerate}
\item If $n\leq M$ and $(x_1,\ldots, x_n)$ is a deadlock, then $(x_1,\ldots,x_n,\top,\ldots,\top)\in T^M$ is a deadlock, by Lem~\ref{lem:topdeadlock}

\item If $n>M$, by Thm.~\ref{thm:deadgeneral} there is a deadlock in $T^m$ for some $m\leq M$ and by Lem.~\ref{lem:topdeadlock} there is a deadlock in $T^M$

\end{enumerate}
\end{proof}

The bound $M$ is sharp in the following sense:
\begin{theorem}\label{thm:deadsharp} For any set of resources $\mathcal{R}$ and capacity function $\kappa:\mathcal{R}\to\mathbb{N}$, $M=\Sigma_{r\in\mathcal{R}}\kappa(r)$, there is a thread $T$ using resources from $\mathcal{R}$, such that $T^M$ has a deadlock and $T^n$ has no deadlock for $n<M$
\end{theorem}
\begin{proof}
Let $T=Pr_1Pr_2Vr_1Pr_3Vr_2\ldots. Pr_kVr_{k-1}Pr_1Vr_kVr_1$

Then the following holds:
 \begin{itemize}
\item There is a deadlock in $T^M$
\item There are no deadlocks in $T^n$ for $n\leq M$
\end{itemize}

The deadlock is at $\mathbf{x}=(\overbrace{x_1,\ldots,x_1}^{\kappa(r_k)},\overbrace{x_2,\ldots,x_2}^{\kappa(r_1)},\ldots,\overbrace{x_k\ldots x_k}^{\kappa(r_{k-1})})=(\mathbf{x_1},\mathbf{x_2},\ldots,\mathbf{x_k})$ where, if we number the $2k+2$ $PV$ steps in $T$ from $1$ to $2k+2$
 \begin{itemize}
\item For $i\neq 1$, $x_i=2i-2$, so $x_i=Pr_i$ and $x_i$ is repeated $\kappa(r_{i-1})$ times and hence holds $\kappa(r_{i-1})$ locks on $r_{i-1}$
\item $x_1=2k$ is the last of the two calls of $Pr_1$ and is repeated $\kappa(r_k)$ times. Holds $\kappa(r_k)$ locks on $r_k$
\end{itemize}

\emph{$\mathbf{x}$ is a deadlock:}

For $i\neq k$, the threads at $x_i$ request a resource which is held by the $\kappa(r_i)$ threads, which are at $x_{i+1}$. The threads at $x_k$ request $r_k$ which is held by the $\kappa(r_k)$ threads at $x_1$

\emph{$\mathbf{x}$ is reachable from $\mathbf{0}$:}

A directed path is composed by $\gamma_0\star\ldots\star\gamma_{k-1}$, where $\gamma_i$ is as follows:

$\gamma_0:\mathbf{0}\to (\mathbf{x_1},\mathbf{0})$ serially - one coordinate at a time. Notice that $\mathbf{0}$ denotes both the $0$-vector with $M$ coordinates, and in  in  $(\mathbf{x_1},\mathbf{0})$ it indicates that the last  $M-\kappa(r_k)$ coordinates are all $0$. $\gamma_0$ is serial execution of $Pr_1Pr_2Vr_1Pr_3Vr_2\ldots,Pr_kVr_{k-1}Pr_1$, $\kappa(r_k)$ times.  $\rho_{r_k}(\mathbf{x_1},\mathbf{0})=\kappa(r_k)$.


For $j\neq 0$,  $\gamma_j:(\mathbf{x_1},\mathbf{0},\mathbf{x_{k-j+2}},\ldots,\mathbf{x_{k}})\to(\mathbf{x_1},\mathbf{0},\mathbf{x_{k-j+1}},\ldots,\mathbf{x_{k}})$ serially. I.e., $Pr_1Pr_2Vr_1Pr_3Vr_2\ldots. Pr_{k-j+1}$ executed $\kappa(r_{k-j})$ times while there are $\kappa(r_i)$ locks on all resources with $i\geq k-j+2$, but none of those are requested by these executions.  Now all resources $r_{k-j}, r_{k-j+1},\ldots, r_k$ are locked to their full capacity.

\emph{There are no deadlocks in $T^n$ for $n<M$:} 

Suppose $\mathbf{y}=(y_1,\ldots,y_n)$ is a deadlock. Let $y_i\neq\top$, i.e., $y_i=x_{j(i)}=Pr_{j(i)}$. There are $\kappa(r_{j(i)})$ threads locking $r_{j(i)}$ and they all have to request a resource. Hence, if $j(i)<k$, there are $l_i=\kappa(j(i))$ threads $y_{j_1}=\ldots =y_{j_{l_i}}=P(j(i)+1)$. Consequently, there are $\kappa(r_{j(i)+1})$ threads holding $r_{j(i)+1}$. If $j(i)=k$, then $\kappa(r_k)$ threads are at $x_1=Pr_1$. Consequently, $\mathbf{y}$ is a permutation of $(\overbrace{x_1,\ldots,x_1}^{\kappa(r_k)},\overbrace{x_2,\ldots,x_2}^{\kappa(r_1)},\ldots,\overbrace{x_k\ldots x_k}^{\kappa(r_{k-1})},\top,\ldots,\top)$, so $n\geq M$  
\end{proof}
\begin{example}  $T=Pa.Pb.Va.Pc.Vb.Pa.Vc.Va$, $\kappa\equiv 1$. $T^3$ has a deadlock at $(6,2,4)$. The path to the deadlock runs  $(0,0,0) \to (6,0,0) \to (6,0,4)\to (6,2,4)$. Care is needed when providing the path as other piecewise serial paths such as $(0,0,0)\to (0,2,0) \to (6,2,0)\to (6,2,4)$ go through the forbidden states - this particular path locks $a$ above its capacity at the point $(2,2,0)$. There are deadlocks at $(6,4,2)$, $(4,6,2)$, $(4,2,6)$, $(2,6,4)$ and $(2,4,6)$. See Fig.~\ref{fig:deadinthree}.
\begin{figure}
\begin{center}
\begin{subfigure}{.45\textwidth}
\begin{tikzpicture}[scale=0.5]
\draw [->] (0,0)--(9,0);
\draw [->] (0,0)--(0,9);

\filldraw[fill=white!20!blue,opacity=.8] (2,2) rectangle (5,5);
\filldraw[fill=white!30!blue,opacity=.8,opacity=.8,draw=black] (1,1) rectangle (3,3);
\filldraw[fill=white!30!blue,opacity=.8,draw=black] (6,1) rectangle (8,3);
\filldraw[fill=white!30!blue,opacity=.8,draw=black] (1,6) rectangle (3,8);
\filldraw[fill=white!30!blue,opacity=.8,draw=black] (6,6) rectangle (8,8);
\filldraw[fill=white!40!blue,opacity=.8,draw=black] (4,4) rectangle (7,7);
\draw (1,0) node[anchor=north]{$Pa$};
\draw (2,0) node[anchor=north]{$Pb$};
\draw (3,0) node[anchor=north]{$Va$};
\draw (4,0) node[anchor=north]{$Pc$};
\draw (5,0) node[anchor=north]{$Vb$};
\draw (6,0) node[anchor=north]{$Pa$};
\draw (7,0) node[anchor=north]{$Vc$};
\draw (8,0) node[anchor=north]{$Va$};

\draw (0,1) node[anchor=east]{$Pa$};
\draw (0,2) node[anchor=east]{$Pb$};
\draw (0,3) node[anchor=east]{$Va$};
\draw (0,4) node[anchor=east]{$Pc$};
\draw (0,5) node[anchor=east]{$Vb$};
\draw (0,6) node[anchor=east]{$Pa$};
\draw (0,7) node[anchor=east]{$Vc$};
\draw (0,8) node[anchor=east]{$Va$};
\draw (2,2) node{$A$};
\draw (7,7) node{$A$};
\draw (7,2) node{$A$};
\draw (2,7) node{$A$};
\draw (3.5,3.5) node{$B$};
\draw (5.5,5.5) node{$C$};
\draw (9,0) node[anchor=west]{$T$};
\draw (0,9) node[anchor=south]{$T$};
\foreach \x in {1,...,8}
\draw (\x cm,1pt)--(\x cm,-1pt);
\foreach \y in {1,...,8}
\draw (-1pt,\y cm)--(1pt,\y cm);

\end{tikzpicture}
\caption{$T^2$ has no deadlocks, but $T^3$ does.}\label{fig:deadinthree}
\end{subfigure}
\hspace{2mm}
\begin{subfigure}{.45\textwidth}
\begin{tikzpicture}[scale=0.4]
\draw [->] (0,0)--(11,0);
\draw [->] (0,0)--(0,11);

\filldraw[fill=white!20!red,opacity=.8] (3,3) rectangle (4.5,4.5);
\filldraw[fill=white!30!blue,opacity=.8,opacity=.8,draw=black] (1,1) rectangle (3.5,3.5);
\filldraw[fill=white!30!blue,opacity=.8,draw=black] (7.5,1) rectangle (10,3.5);
\filldraw[fill=white!30!blue,opacity=.8,draw=black] (1,7.5) rectangle (3.5,10);
\filldraw[fill=white!30!blue,opacity=.8,draw=black] (7.5,7.5) rectangle (10,10);
\filldraw[fill=white!40!magenta,opacity=.8,draw=black] (4,4) rectangle (7,7);
\filldraw[fill=white!40!red,opacity=.8,draw=black] (6.5,6.5) rectangle (8,8);
\filldraw[fill=white!40!red,opacity=.8,draw=black] (3,6.5) rectangle (4.5,8);
\filldraw[fill=white!40!red,opacity=.8,draw=black] (6.5,3) rectangle (8,4.5);
\filldraw[fill=white!40!green,opacity=.8,draw=black] (1.5,1.5) rectangle (2.5,2.5);
\filldraw[fill=white!40!green,opacity=.8,draw=black] (8.5,8.5) rectangle (9.5,9.5);
\filldraw[fill=white!40!green,opacity=.8,draw=black] (1.5,8.5) rectangle (2.5,9.5);
\filldraw[fill=white!40!green,opacity=.8,draw=black] (8.5,1.5) rectangle (9.5,2.5);
\filldraw[fill=white!40!green,opacity=.8,draw=black] (1.5,5) rectangle (2.5,6);
\filldraw[fill=white!40!green,opacity=.8,draw=black] (5,1.5) rectangle (6,2.5);
\filldraw[fill=white!40!green,opacity=.8,draw=black] (8.5,5) rectangle (9.5,6);
\filldraw[fill=white!40!green,opacity=.8,draw=black] (5,8.5) rectangle (6,9.5);
\filldraw[fill=white!40!green,opacity=.8,draw=black] (5,5) rectangle (6,6);

\draw (11,0) node[anchor=west]{$T$};
\draw (0,11) node[anchor=south]{$T$};

\end{tikzpicture}
\caption{ $T^2$ is serializable.}\label{fig:serializable}
\end{subfigure}
\end{center}
\caption{ a)  There are 18 forbidden rectangles in the model of $T^3$, namely 3 for each of the blue rectangles pictured here. As in Fig.~\ref{fig:serialize}\newline b)  $T=Pa.Pd.Pb.Va.Pc.Vb.Pd.Vd.Pb.Vc.Pa.Vb.Pd.Vd.Va$. In the green area, there are two locks on $d$, in the red area on $b$, in the blue on $a$ and the magenta on $c$. }
\end{figure}
\end{example}

\section{Serializability}
An execution is serializable if it is equivalent to a serial execution - one thread is executed from $\perp$ to $\top$ at a time. A program is serializable if all the executions of it are serializable, Def.~\ref{def:serializable}. For a thread $T$ calling only resources of capacity $\kappa=1$, $T^n$ is serializable if and only if $T^2$ is serializable, Thm.~\ref{thm:serializabilitykapacity2}.  In Thm.~\ref{thm:serializableiffconnected} we prove that when all resources have capacity at least 2, a program is serializable if and only if all executions are equivalent. 

Moreover, there is an obstruction to serializability, a property of individual states such that if no such obstruction exists, then all executions are equivalent. There are such obstructions in $T^n$ if and only if there are obstructions in $T^M$ where $M=\Sigma_{r\in\mathcal{R}}\kappa(r)+1$. Hence, for capacity at least $2$, if there are no such obstructions in $T^M$, then $T^n$ is serializable for all $n$.

\begin{definition}\label{def:serializable}

Consider a PV-program $p=T1|T2|\cdots |Tn$ with geometric model $X$. An execution $\gamma:I\to X$, $\gamma(0)=\perp$, $\gamma(1)=\top$ is serial if there is a subdivision $0=t_0<t_1<t_2\cdots <t_n=1$ of $[0,1]$ and a bijection $\sigma:\{1,\ldots,n\}\to\{1,\ldots,n\}$ such that for $t\in [t_{i-1},t_i]$, $\gamma_j(t) \in\{\perp,\top\}$ for $j\neq \sigma(i)$ and moreover, $\gamma_{\sigma(i)}(t_{i-1})=\perp$ and $\gamma_{\sigma(i)}(t_i)=\top$.

An execution $\mu$ is \emph{serializable} if there is a serial execution $\gamma:I\to X$ which is equivalent to $\mu$ in the sense of Def.~\ref{def:execution}.

The program $p$ is serializable if all executions are serializable. 
\end{definition}
 \subsection{Capacity $1$, mutexes.}
\begin{theorem}\label{thm:serializabilitykapacity2}
Let $T$ be a PV-thread acquiring only resources of capacity 1. Then $T^n$ is serializable for all $n$ if and only if $T^2$ is serializable.
\end{theorem}
 In \cite{LipPap81} and \cite{Yannakakis1979}  which are also geometric, they study serializability (which they call safety) for the general case $T1|T2|\ldots|Tn$ and prove that $Ti|Tj$ has to be serializable for all pairs and moreover, there is a condition on a cycle in a graph. Furthermore, in \cite{LipPap81} they give an example of three processors which are pairwise serializable, but $T1|T2|T3$ is not serializable. Hence, Thm.~\ref{thm:serializabilitykapacity2} does not hold when the threads are different. 
Our proof relies on  the algorithms and results on classification of executions of a simple $PV$-program up to equivalence, see e.g. \cite{Fajstrup2016} p.130 and \cite{Raussen2010b}. 

 The part needed here is given in the following, where we use the model from Def.~ \ref{def:geommodel}  in the case where all resources have capacity $1$:
 \begin{remark} with notation from Def.~\ref{def:geommodel} where $F_p=\bigcup_{l=1}^N R^l$ and $R^l=I_1^l\times\ldots I_n^l$ corresponds to a resource $r(l)$ with $\kappa(r(l))=1$ there are two indices $K^l=\{k^l_1,k^l_2\}$  s.t. $I_j=[0,l_j+1]$ for $j\notin K^l$ and $I_{k_j}=]a^l_{k_j},b^l_{k_j}[$. A given resource may give rise to more than one rectangle. 

At $R^l$, $T_{k_i}$ requests $r(l)$ at $a_{k_i}^l$. $R^l$ represents one conflict  involving the resource $r(l)$ and these two threads. 
 \end{remark}
A schedule is a geometric construction in \cite{Raussen2010b} which for each rectangle specifies a thread which does not get the resource last. With mutexes, this simply states who gets it first. The following gives that definition in this simple case 
\begin{definition}\label{def:schedule}[Schedules]
A schedule $S$ for a  PV-program is a choice for each rectangle $R^l$ of one of the two non-trivial directions, $s_l\in\{k^l_1,k^l_{2}\}$ - we  call the other one $t_l$. An execution path obeys the schedule $s_1,\ldots, s_m$, if it does not intersect any of the \emph{extended} rectangles $R^l_ {s_l}=\times_{k=1}^n I^{l,s_ l}_k$, where $I^{l,s_ l}_k=[0,b^l_k[$ for $k=s_l$ and $I^{l,s_ l}_k=I^l_k$ otherwise. The thread $Ts_l$ gets the lock at step $a^l_{s_l}$ on the resource $r(l)$ before $Tt_l$ gets the lock at $a^l_{t_l}$.
\end{definition}

\begin{example}\label{ex:schedule} An execution path $\gamma: I\to X$ where $R=I_1\times I_2 \times \cdots \times I_n$, $I_j=]a_j,b_j[$, $I_i=]a_i,b_i[$, $j\neq i$ and $I_k=[0,l_k+1]$ otherwise will satisfy
\begin{itemize}

\item $\gamma_i^{-1}(]a_i,b_i[)\cap \gamma_j^{-1}( ]a_j,b_j[)=\emptyset$ since the resource  has capacity $1$, i.e.,  if $\gamma_i(t)\in ]a_i,b_i[$ and $\gamma_j(t)\in ]a_j,b_j[$, then $\gamma(t)\in R$, which is forbidden.
\item Hence, either
\begin{enumerate}
\item $\gamma_i^{-1}(]a_i,b_i[)< \gamma_j^{-1}( ]a_j,b_j[)$ or
\item $\gamma_i^{-1}(]a_i,b_i[)> \gamma_j^{-1}( ]a_j,b_j[)$. 
\end{enumerate}
\end{itemize}
A schedule at $R$ is a choice between 1) and 2).
\end{example}

\begin{example}\label{ex:serializable} With  $T=Pa.Pd.Pb.Va.Pc.Vb.Pd.Vd.Pb.Vc.Pa.Vb.Pd.Vd.Va$, not all schedules for $T^2$ are allowed. See Fig.~\ref{fig:serializable}.  Either at each rectangle the choice of an execution is right-then-up or it is up-then-right at all of them. Even though some of the green rectangles are disconnected from the others. $T^2$ is serializable. And by Thm.~\ref{thm:serializabilitykapacity2}, so is $T^n$ for general $n$.
\end{example}
There may be no executions obeying a given schedule, as in Ex.~\ref{ex:serializable}, but we have: 
\begin{proposition}\label{prop:schedulingalgo}(a special case of \cite{Fajstrup2016} Prop.~7.9)
Let $S$ be a schedule for a PV-program. All executions obeying $S$ are equivalent. Moreover, for all executions $\gamma$, there is a set of schedules which are obeyed by $\gamma$. If all resources have capacity $1$, the schedule is unique.
\end{proposition}
For general capacities, there may be executions obeying more than one schedule. 
\begin{proof} For existence, see \cite{Fajstrup2016} Prop.~7.27. The upshot is that an execution will contain choices for all the rectangles - it is a resolving of the conflicts.

Uniqueness: Suppose $S\neq \hat{S}$ are different schedules obeyed by $\gamma$. Let $R$ be a rectangle, where $s\neq \hat{s}$. With notation from  Ex. \ref{ex:schedule}  this means that $s=i$ and $\hat{s}=j$ (or vice versa)


If $\gamma$ obeys $S$, then $\gamma_j^{-1}( ]a_j,b_j[)$ is less than and not intersecting  $\gamma_i^{-1}( ]a_i,b_i[)$. To obey $\hat{S}$, the inequality is reversed. These executions/paths are not equivalent. The order of the access to the resource is reversed, which for resources of capacity at least $2$ may not imply inequivalence, but with capacity $1$, it does: A continuous deformation from one to the other is not possible without intersecting $R$. We will not spell that out - again, it is in  \cite{Fajstrup2016} p.130 and \cite{Raussen2010b}. 
\end{proof}

\begin{lemma}\label{lem:serializablemutex} Let $T=Pr.Vr$, $\kappa(r)=1$ and consider $p=T^n$. There are $n!$ equivalence classes of executions and they all contain a serial path. Moreover, all serial paths are inequivalent.
\end{lemma}
\begin{proof}
The geometric model is $X=[0,4]^n\setminus F$ where $F$ consists of $n\cdot(n-1)/2$ forbidden rectangles: For $i<j$, $R_{ij}=\times_{k=1}^nI_k$ $I_i=I_j=]1,2,[$.

A schedule is a choice of $i$ or $j$ for each rectangle, i.e., whether $Ti$ or $Tj$ passes $]1,2[$ first. This pairwise order gives a total order - i.e., there are no "loops" $i<j<\ldots <i$, since a dipath $\gamma$ obeying such a schedule would satisfy  $\gamma_i^{-1}(]1,2[)<\gamma_j^{-1}(]1,2[)<\gamma_i^{-1}(]1,2[)$.

Hence, the set of schedules is in bijection with the set of total orders on $\{1,\ldots,n\}$. Each order is obeyed by a serial execution - following the order. By Prop.~\ref{prop:schedulingmutex}  a schedule determines an equivalence class of executions uniquely. I.e., all executions are equivalent to a serial execution and no serial executions are equivalent.
\end{proof}
 Notice that since all serial executions consist of executing $T$ $n$ times, they may seem to be the same execution. They are not, however. They differ in the order in which each PV-thread gets access to resources.  

\begin{proof}[Of Thm.~\ref{thm:serializabilitykapacity2}]. 
Suppose $T$ is non-trivial, otherwise all $T^n$ are serializable.
Suppose now that $T^2$ is serializable. The proof consist of proving
\begin{enumerate}
\item All the $n!$ serial executions of $T^n$ are inequivalent - there is no directed homotopy between them.
\item There are at most $n!$ equivalence classes of executions of $T^n$.
\end{enumerate}
Then all equivalence classes of executions has to contain a serial execution and therefore $T^n$ is serializable.

Proof of 1):  Let $r$ be a resource requested by $T$ in the interval $]a,b[$. The $n(n-1)/2$ rectangles as in Lem.~\ref{lem:serializablemutex} ensure that all $n!$ serial executions are  inequivalent.

Proof of 2):  Suppose $T$ holds the resource $r$ at intervals $]a^{r,l},b^{r,l}[,\;l=1,\ldots,m_r$. The forbidden area for $T^2$ is the union of all rectangles $]a^{r,l},b^{r,l}[ \times ]a^{r,k},b^{r,k}[$, $l,k\in [1:m_r]$  for all $r$.

Prop.~\ref{prop:schedulingalgo} implies that at each such rectangle, an execution path has a choice of the two orderings, a schedule is such a choice and executions are classified by their schedule at all rectangles. 

As $T^2$ is serializable, all executions are equivalent to one of the two serial executions. Hence, only two schedules for $T^2$ give rise to a non-empty set of executions. I.e., the schedule at one rectangle fixes the schedule at all the others. 

The forbidden region in $T^n$ is the union of all $\times_{k=1}^n I_k$ where $I_k=[0,l_k+1]$ for all except two directions, $I_i=]a^{r,l},b^{r,l}[$ and $I_j= ]a^{r,k},b^{r,k}[$, the \emph{non-trivial directions} for that rectangle.
Let $S$ be a schedule for an execution $\gamma$ of $T^n$
 For a fixed pair $i,j$, $S$ in particular gives a schedule  for all rectangles with $i,j$ as the non-trivial directions. Since $\gamma$ does not intersect any of these rectangle, $(\gamma_i,\gamma_j)$ is an execution of $Ti|Tj$, i.e. of $T^2$ with the schedule given by $S$. Hence, either  $i$ goes first at all these rectangles or $j$ does. Consequently, a schedule for $T^n$ is a choice for each pair $i,j$ of an order. 

\emph{These pairwise orders give a total order:}
Let $]a,b[$ be an interval where $T$ holds a resource $r$. The schedule given by the pairwise order in particular is a schedule for the $n(n-1)/2$  rectangles associated with this. Hence,  the pairwise order gives rise to a total order and there are $n!$ of those.

All in all, a schedule is a total order on $\{1,\ldots,n\}$ , all executions are equivalent to a serial execution, and none of these are equivalent.

\end{proof}
The examples in \cite{LipPap81} of pairwise serializability and $T1|T2|T3$ not serializable is precisely such a "loop", where the pairwise order does not give a total order - this is possible when each resource is accessed by two threads, not by all three.

\begin{example}If resources of capacity 1 are mixed with higher capacity resources, Thm.~\ref{thm:serializabilitykapacity2} is not true. Let $T1=PaPdPbVbPcVcVdVa$, $T2=PaPbPcVaPdVdVbVc$, $T3=PaPbVbVaPcPdVdVc$ and suppose $\kappa(a)=\kappa(b)=\kappa(c)=2$. Then $p=T1|T2|T3$ is non serializable, this is "two wedges" in Fig.~ \ref{fig:localchoice}- and \cite{Raussen2006}. Let $r$ be a resource of capacity $1$. Then with $T=Pr.Vr.T1.T2.T3$, $T^2$ is serializable, but $T^3$ is not.
\end{example}
\subsection{Capacity at least $2$, general semaphores}
For higher capacities, the geometry is very different:
\begin{theorem}\label{thm:serializableiffconnected}
Let $T1,\ldots,Tn$ be valid PV-threads, calling resources of capacity at least 2. Then $T1|\ldots|Tn$ is serializable if and only if all executions are equivalent.
\end{theorem}
\begin{proof}
It suffices to see that all serial executions are equivalent.

For that, let $\gamma=T1.T2.\ldots.Tn$ be a serial execution. It suffices to see that $\gamma$ is equivalent to $T1.T2.\ldots Ti-1.Ti+1.Ti\ldots.Tn$, where two threads have been interchanged. Let $[t_0,t_1]=\gamma^{-1}(Ti.Ti+1)$, then $$\gamma(t_0)=(\overbrace{\top,\ldots,\top}^{i-1},\perp,\ldots,\perp);\;\gamma(t_1)=(\overbrace{\top,\ldots,\top}^{i+1},\perp,\ldots,\perp)$$

any execution path from $\gamma(t_0)$ to $\gamma(t_1)$ is an executiion of $Ti|T(i+1)$. All resources have capacity at least $2$, so no states are forbidden in $Ti|T(i+1)$. By Ex.~\ref{ex:elemhtpy} all execution paths are equivalent. In particular $Ti.T(i+1)$ and $T(i+1).Ti$
\end{proof}

In \cite{MRMSCS}, a sufficient condition for serializability may be found. The setting there is  more general than what we need, so we spell it out in the special case of a $PV$-program with a geometric model as in Def.~\ref{def:forbidden}.

In \cite{MRMSCS} Prop. 2.18, conditions are given for when all directed paths are dihomotopy equivalent. We do not give these conditions in general. For our purpose, the following more specific formulation suffices:
\begin{proposition}\label{prop:di1connected}(From \cite{MRMSCS} Prop.~2.20 and 2.18) Let $X$ be as in  Def.~\ref{def:forbidden}. All pairs of dipaths $\gamma$, $\mu$ in $X$ with $\gamma(0)=\mu(0)$ and $\gamma(1)=\mu(1)$ are dihomotopic if the following condition, $\star$, holds: 
For every point $\mathbf{x}\in X$ and every pair of \emph{edges} $(x_1,x_2,\ldots,x_l+t,\ldots,x_n)$ and $(x_1,x_2,\ldots,x_m+t,\ldots,x_n)$, $t\in[0,\varepsilon[$, $\varepsilon>0$ which are in $X$,
there is an  $\hat\varepsilon>0$ s.t. $\hat\varepsilon\leq \varepsilon$,  and a  sequence ${i_j}$, $j=1,\ldots,k$ such that $(x_1,x_2,\ldots,x_{i_j}+t,\ldots,x_n)$ is in $X$ for $t\in[0,\hat\varepsilon[$, $i_1=l$ and $i_k=m$ and all pairwise connections by rectangles 
$(x_1,x_2,\ldots,x_{i_j}+u,\ldots,x_{i_j+1}+v,\ldots,x_n)$ in $X$ for $(u,v)\in [0,\hat\varepsilon[\times [0,\hat\varepsilon[$. 
\end{proposition}
The intuition is that there is no local choice of a direction from $\mathbf{x}$. If a dipath proceeds along direction $l$ another along $m$, they are connected via the rectangles in the condition.

This leads to the following definition of obstructions to all dipaths and in particular all partial executions being dihomotopic:
\begin{definition}\label{def:localchoicepoint}
A $PV$-program has a \emph{local choice point} at $\mathbf{x}=(x_1,\ldots,x_n)$ if 
\begin{enumerate}
\item there is a resource $\tilde{r}$ and a subset $S=\{ i_1,\ldots,i_m\}\subset[0:n]$, $m\geq 2$, such that  all corresponding threads $Ti_j$ request $\tilde{r}$ at $\mathbf{x}$, i.e.,  $x_{i_j}=P\tilde{r}$,
\item $\rho_{\tilde{r}}(x)=\kappa(\tilde{r})-1$,
  
\item for $i\not\in S$  either $x_i=\top$ or $x_i=Pr$ for an $r\in\mathcal{R}$ with $\rho_r(\mathbf{x})=\kappa(r)$
\end{enumerate}
\end{definition}
\begin{remark} There are "dual" choice points $\mathbf{y}$ where  there is a resource $\tilde{r}$ and a subset $S=\{ i_1,\ldots,i_m\}\subset[0:n]$, $m\geq 2$, such that  all corresponding threads release $\tilde{r}$ at $\mathbf{y}$, i.e.,  $y_{i_j}=V\tilde{r}$,  $\rho_{\tilde{r}}(x)=\kappa(\tilde{r})-1$,  for $i\not\in S$  either $y_i=\perp$ or $y_i=Vr$ for an $r\in\mathcal{R}$ with $\rho_r(\mathbf{x})=\kappa(r)$. These are choices "going backwards" and could be used as choice points too.
\end{remark}

\begin{theorem}\label{thm:localchoicepoint}
For a  $PV$-program with no local choice points all executions are equivalent.
\end{theorem}
\begin{proof}
We will  see that the condition $\star$ in Prop.~\ref{prop:di1connected} is violated  if and only if there is a local choice point. 

\emph{A local choice point gives a violation of $\star$:} Suppose $\mathbf{x}$ is a local choice point. The threads $T_{i_j}$ where $i_j\in S$ are the threads which may proceed - all others are blocked either at $\top$ or by requesting a locked resource. In Prop.~\ref{prop:di1connected} these are the directions ${i_j}$ s.t.  $(x_1,x_2,\ldots,x_{i_j}+t,\ldots,x_n)$ is in $X$ for $t\in[0,\varepsilon[$. Since only one of these threads may proceed - it will lock the resource $\tilde{r}$, which is then locked to its full capacity - any rectangle $(x_1,x_2,\ldots,x_{i_j}+u,\ldots,x_{i_j+1}+v,\ldots,x_n$ in $X$ for $(u,v)\in [0,\varepsilon[\times [0,\varepsilon[$ will intersect $F$: when $u> 0$ and $v> 0$ such a state corresponds to having $\kappa(\tilde{r})+1$ locks on $\tilde{r}$.

\emph{If the condition $\star$ is violated, then there is a local choice point:} Let  $\mathbf{x} \in X$ be a point where $\star$ does not hold. Then  there is an $\varepsilon>0$ and $l\neq m$ such that the dipaths  $(x_1,x_2,\ldots,x_l+t,\ldots,x_n)$ and $(x_1,x_2,\ldots,x_m+t,\ldots,x_n)$, $t\in[0,\varepsilon[$,  are in $X$, but for all $\hat\varepsilon <\varepsilon$ and sequences of allowed edges $(x_1,x_2,\ldots,x_{i_j}+t,\ldots,x_n)$ in $X$ for $t\in[0,\hat\varepsilon[$,  ${i_j}$, $j=1,\ldots,k$ $l=i_1,i_2,\ldots,i_k=m$ there is at least one connecting rectangle which intersects $F$. 

In particular, the sequence $l,m$ has this property, so for all $\hat\varepsilon$, $ \{(x_1,x_2,\ldots,x_l+u,\ldots,x_m+v,\ldots,x_n) \mid u,v\in]0,\hat\varepsilon[\} \cap F\neq\emptyset$.  Hence, there is a resource $\tilde{r}$ such that $\rho_{\tilde{r}}(\mathbf{x})\leq\kappa(\tilde{r})$, $\rho_{\tilde{r}} (x_1,x_2,\ldots,x_l+u,\ldots,x_m+v,\ldots,x_n) >\kappa(\tilde{r})$ for $u>0$ and $v>0$. This implies that $x_l$ or $x_m$ is $P\tilde{r}$ and consequently $\rho_{\tilde{r}}( (x_1,x_2,\ldots,x_l+u,\ldots,x_n) )=\rho_{\tilde{r}}(\mathbf{x})+1$ or $\rho_{\tilde{r}} (x_1,x_2,\ldots,x_l+v,\ldots,x_n) =\rho_{\tilde{r}}(\mathbf{x})+1$. Moreover, since $Tl$ and $Tm$ may proceed,  $\rho_{\tilde{r}} (x_1,x_2,\ldots,x_l+u,\ldots,x_n)\leq\kappa(\tilde{r})$ and $\rho_{\tilde{r}} (x_1,x_2,\ldots,x_l+u,\ldots,x_n)\leq\kappa(\tilde{r})$. Hence, $\rho_{\tilde{r}}(\mathbf{x})=\kappa(\tilde{r})-1$, $x_l=x_m=P\tilde{r}$ and $\rho_{\tilde{r}} (x_1,x_2,\ldots,x_l+u,\ldots,x_m+v,\ldots,x_n)=\kappa(\tilde{r})+1$.
 
 This holds for all allowed directions $(x_1,x_2,\ldots,x_{i_j}+t,\ldots,x_n)$. Let $S$ be the set of these $i_j$. They then all request access to $\tilde{r}$. In particular, \emph{all} connecting rectangles $(x_1,x_2,\ldots,x_{i_j}+u,\ldots,x_{i_k}+v,\ldots,x_n)$  intersect $F$.

 For $k\not\in S$, $(x_1,x_2,\ldots,x_{k}+t,\ldots,x_n)$  is not in $X$ for small positive $t$, hence  $x_k=\top$ or $ (x_1,x_2,\ldots,x_{k}+t,\ldots,x_n)  \in F$, i.e.,   $x_k=Ps$ with $\rho_s(\mathbf{x})=\kappa(s)$. Hence, $\mathbf{x}$ is a local choice point.

\end{proof}

\begin{example} In Fig.~\ref{fig:localchoice} b) and Example~\ref{ex:localchoice},  the local choice point is at $(2,2,2)$, the resource $\tilde{r}$ in Def.~\ref{def:localchoicepoint} is $d$ and $S=\{2,3\}$. $T1$ cannot proceed, sice $\rho_c(2,2,2)=2$. 
\end{example}
\begin{theorem}\label{thm:choicegeneral} Let $p=T1|T2|\ldots|Tn$ be a valid PV-program, such   that $\kappa(r)\geq 2$. Suppose $n\geq M$ where $M={\Sigma_{r\in\mathcal{R}}\kappa(r)}+1$. If for all subsets $\{i_1,i_2,\ldots, i_m\}\in [1:n]$ of $M$ threads, there are no local choice points in $Ti_1|\ldots|Ti_m$, then there are no local choice points in $p$. 
\end{theorem}
\begin{proof} 

Let $\mathbf{x}=(x_1,\ldots,x_n)\in p$ be a local choice point and suppose $x_m$ and $x_l$ are as in the proof of \ref{thm:localchoicepoint}.

The following construction gives  $\{i_1,i_2,\ldots i_k\}\subset[1:n]$ and a  local choice point in $Ti_1|\ldots|Ti_k$ with $k\leq M$.

Choose $T_{i_1},\ldots,T_{i_k}$  as follows: If $x_j$ does not hold a resource and $j\notin \{ l,m\}$, then omit $T_j$. Let $\hat{\mathbf{x}}=(x_{i_1},\ldots,x_{i_k})$ be the resulting point in $Ti_1|\ldots|Ti_k$. Then $\hat{\mathbf{x}}$ is a local choice point: 
 First, notice that $\rho_r(\mathbf{x})=\rho_r(\hat{\mathbf{x}})$, since no resources were held by the omitted threads,

\begin{enumerate}
\item All $\hat{x}_j$ are access requests, since  if $x_j=\top$ it does not hold a resource and it is not in $\hat{\mathbf{x}}$.
\item Since $x_l$ and $x_m$ are not omitted, $x_l=x_m=P_{\tilde{r}}$ still holds  and $\rho_{\tilde{r}}(\hat{\mathbf{x}})=\kappa(\tilde{r})-1$ 
\item If $x_i=P{r}$ and $r\neq\tilde{r}$ then $\rho_r(\hat{\mathbf{x}})=\kappa(r)$, since  $\rho_r(\mathbf{x})=\rho_r(\hat{\mathbf{x}})$
\end{enumerate}

$k\leq M$: For all $i\notin \{ l,m\}$, $\hat{x}_i$ holds a resource and $\rho_{\tilde{r}}(\hat{\mathbf{x}})=\kappa(\tilde{r})-1$. Hence, $k\leq\Sigma_{r\in\mathcal{R}}\kappa(r)-1+2=\Sigma_{r\in\mathcal{R}}\kappa(r)+1=M$. 

There is then a choice point in a subset of size $M$, $Ti_1|\ldots|Ti_k|Tj_1|\ldots |Tj_{M-k}$ at $(x_{i_1},\ldots,x_{i_k},\top,\ldots,\top)$.
 
\end{proof}

\begin{theorem}\label{thm:cut-offchoicepoint} If there are no local choice points in $T^M$ where $M=1+\Sigma_{r\in\mathcal{R}}\kappa(r)$, then there are no local choice points in $T^n$ for any $n$.

If at least one resource has capacity $1$, there is a choice point in $T^n$ for all $n\geq 2$.

If $T^n$ is non-serializable for some $n$, then there are local choice points in $T^M$. If $T^M$ has no local choice points, then $T^n$ is serializable for all $n$.

\end{theorem}
\begin{proof} The last statement is a consequence of Prop.~\ref{prop:di1connected}. 

If $\kappa(\tilde{r})=1$, then any point $(P\tilde{r},P\tilde{r},\top,\ldots,\top)$ is a choice point.

 For  $\kappa\geq 2$, the argument in the proof of \ref{thm:choicegeneral} works. A choice point in $T^n$ gives a choice point in $T^M$ either by adding $\top$ at the remaining coordinates or by reducing to $\hat{x}$ as in \ref{thm:choicegeneral}.

\end{proof}
When $\kappa =1$, the bound is $2$ and is certainly sharp. The bound is not known to be sharp for general capacity. More precisely:
\begin{theorem}\label{thm:sharpserializable}
For any set of resources $\mathcal{R}$ and capacity function $\kappa:\mathcal{R}\to\mathbb{N}\setminus \{1\}$, $M=\Sigma_{r\in\mathcal{R}}\kappa(r)+1$ there is a thread $T$ using resources from $\mathcal{R}$, such that $T^M$ has a reachable local choice point and  $T^n$ has no local choice point for $n\leq M-2$.

\end{theorem}
\begin{proof}
 Let $\mathcal{R}=\{r_1,\ldots,r_k\}$
 and $T=Pr_1Pr_2Vr_1Pr_3Vr_2\ldots,Pr_kVr_{k-1}Pr_1Vr_kVr_1$
as in Thm.~\ref{thm:deadsharp}. 

Let $\tilde{r}=r_k$. There is a choice point in $T^{M-2}$ at \[\mathbf{x}=(\overbrace{x_1,\ldots,x_1}^{\kappa(r_k)-1},\overbrace{x_2,\ldots,x_2}^{\kappa(r_1)},\ldots,\overbrace{x_k\ldots x_k}^{\kappa(r_{k-1})}) \] where $x_i$ are as in Lem.~\ref{thm:localchoicepoint}:  Let  $S$ be the last $\kappa(r_{k-1})$ threads, they request $r_k$ which is held by the first $\kappa(r_k)-1$ threads. All threads not in $S$ are prevented from proceeding as in \ref{thm:localchoicepoint}. $\mathbf{x}$ is reachable by an execution path from $\mathbf{0}$ as in Thm.~\ref{thm:deadsharp}. Hence, $(\mathbf{x},\top,\top)$ is a local choice point in $T^M$.

Suppose $\mathbf{y}=(y_1,\ldots,y_n)$ is a local choice point and that no thread is at $\top$. Let $\tilde{r}=r_i$ be the resource requested with $\rho_{r_i}(\mathbf{y})=\kappa(r_i)-1$. Suppose $i=1$ - for $i\neq 1$, the argument is similar. Suppose after reordering that $y_n,y_{n-1}\in S$. As $r_1$ is held by $\kappa(r_1)-1$ threads, $y_1,\ldots,y_{\kappa(r_1)-1}=x_2=Pr_2$. Hence, $\kappa(r_2)$ threads hold $r_2$, i.e., they are at $x_3$ and request $r_3$ and $\mathbf{y}$ is a permutation of the coordinates in $\mathbf{x}$. In particular $\mathbf{y}\in T^{M-2}$.

\end{proof}

\section{Algorithmic considerations}
The deadlock algorithm in \cite{LFEGMRAlgebraic}  may of course be applied to find the deadlocks at the cut-off, but local choice points are very close to being deadlocks and hence may also be found using the deadlock algorithm:
\begin{definition}\label{def:potentialdead} A potential deadlock in a $PV$-program is a state $\mathbf{x}\in X\cup F_p$ such that 
\begin{itemize}
\item $\mathbf{x} \neq\top$
\item All $x_i$ are access requests $P{r(i)}$ or $x_i=\top$
\item For all  $i\in [1:n]$, s.t. $x_i\neq \top$, $\rho_{r(i)}(x_1,\ldots,x_n)= \kappa(r(i))$
\end{itemize}

\end{definition}

A potential deadlock which is reachable is a deadlock.
\begin{proposition}\label{prop:ifchoicethendead}
Let $\mathbf{x}=(x_1,\ldots,x_n)$ be a local choice point in $T^n$ and suppose all resources have capacity at least $2$. Then there is $k\in [1:n]$ such that $(x_k, x_1,\ldots,x_n)\in T^{n+1}$ is a potential deadlock.
\end{proposition}
\begin{proof}
The condition for a potential deadlock is satisfied for all $i$ except  for at least $2$ where $r(l)=r(m)=\tilde{r}$ and $\rho_{\tilde{r}}(x_1,\ldots,x_n)= \kappa(\tilde{r})-1$
as in the proof of \ref{thm:localchoicepoint}. As $\kappa(\tilde{r})\geq 2$, there is a $k$ such that $\rho_{\tilde{r}}(x_k)=1$. Then $\rho_{\tilde{r}}(x_k,x_1,\ldots,x_n)= \kappa(\tilde{r})$. 
\end{proof}
\begin{corollary} If there are no potential deadlocks in $T^N$, where $N=2+\Sigma_{r\in\mathcal{R}}\kappa(r)$, then $T^n$ is serializable for all $n$.
\end{corollary}
\begin{proof}
Thm.~\ref{thm:cut-offchoicepoint} and Prop.~\ref{prop:ifchoicethendead}. 
\end{proof}
Such a potential deadlock in $T^n$  is well understood in the geometric model: It is a minimum of an intersection of $N$  forbidden $n$-rectangles. To be a true deadlock, it would have to be outside the forbidden region  - i.e., not in the interior of any other rectangles - and also reachable. The deadlock algorithm in \cite{Fajstrup2016} pp.96-97 can be applied for this. 
\begin{remark} The potential deadlock arising from a local choice point may be in the forbidden region - if $x_k$ holds a resource $r\neq\tilde{r}$, which is requested by some $x_i$,  then it is already locked to its full capacity by $\mathbf{x}$ and hence, $\rho_r((x_k,x_1,\ldots,x_n)=\kappa(r)+1$. Hence, Thm. \ref{thm:deadgeneral} does not apply to potential deadlocks.
\end{remark}
\section{Comparison with other work and notions of cut-off.}
 Our use of the term cut-off is related to, but not the same as the cut-off for unfolding of Petri nets in \cite{Esparza}. Their cut-off is a minimal complete prefix which represents all possible unfoldings of loops and branchings.  In \cite{LFMSCS} we prove a result closer to that, namely that there are finite cut-offs for unfolding of nested loops in the $PV$-model, when the aim is to find deadlocks and states from which the program cannot finish.  I.e., the question is how many times one should unfold the loops in a parallel program, where the threads and in particular the number of threads is fixed. The present results are not about unfoldings. Here the number of threads is the parameter. 

In \cite{EmersonNamjoshi95} a cut-off is indeed a number of threads. However, the setting is passing a token around a ring and only one thread may proceed at a time. Here, we allow for general semaphores - resources may have capacity  more than $1$ and also we allow true concurrency. In \cite{NamjoshiTrefler12}, \cite{NamjoshiTrefler16}, \cite{NamjoshiTrefler18} there is a network of threads - the neighbor relation in the network gives the interdependence of threads. The symmetry considerations are based on symmetry of these networks.  Our work here corresponds to the complete graph in that setting in the sense that all threads in $T^n$ will interact pairwise and, in the case of general semaphores, the interaction is not only pairwise.
Our cut-offs allow us, even for our complete graphs, to consider only $M$ threads, where $M$ is the cut-off.  Similarly, our local, neighborhood states will be global in their sense. Our local obstructions are local, not in the sense of considering fewer threads but in the sense of not considering a total execution path. The local properties at a node in the sense of \cite{NamjoshiTrefler18} are related to our local obstructions. Moreover, the compositionality in \cite{NamjoshiTrefler18} is in a way related to our Prop.~\ref{prop:deadcombo}, but Rem.~\ref{rem:compositionality} is a warning against a too simple minded decomposition. 

In \cite{GermanSistla} they provide automatic verification of certain systems with an arbitrary number of a given process run in parallel. They do not allow more than $2$ processes to execute at the same time, and an execution such as the dash-dotted one in Fig.~\ref{fig:serialize} would not be allowed.  Their methods may still be applicable in our setting, but it seems unclear how to handle semaphores which are not mutual exclusion and also in what way their methods would behave under subdivision of the actions.  Our results are more similar to higher dimensional automata, HDA \cite{Pratt1} in the sense of allowing more actions in parallel. The dash-dotted execution however, is not a priori in the HDA built from that $PV$-program. In  the words of V.~Pratt \cite{Pratt}, the $PV$-model allows a continuous change from $a$ before $b$ to $b$ before $a$ (the dihomotopies),  which an HDA does not, so again, care is needed. 

\section{Conclusions and further work.}
For deadlocks, the cut-off at $M=\Sigma_{r\in\mathcal{R}}\kappa(r)$ is sharp, i.e.,  $T^M$ is deadlock free if and only if $T^n$ is deadlock free for all $n$. Serializability is guaranteed if there are no local choice points in $T^M$ with $M=\Sigma_{r\in\mathcal{R}}\kappa(r)+1$. This cut-off is not sharp, but there is a lower bound on the cut-off at $M-2$. Freedom from local choice points is a sufficient condition for serializability, but not necessary. 
The highly symmetric case considered here has not been studied in the geometric setting before and this is just a beginning. In algebraic topology there is a vast literature on symmetry and equivariance, which should be brought into directed topology and be applied to the study of these and other symmetric situations. A group action in the directed topology sense would provide automorphisms which preserve not only executions but equivalences of executions: Directed paths would map to directed paths and directed homotopies - i.e. equivalence of executions - would map to directed homotopies. The symmetry considerations and indeed group actions  in e.g.\cite{EmersonSistla}, \cite{ClarkeEmersonSistla} may be used on top of, or instead of these directed symmetries. In Rem.~\ref{rem:compositionality} such symmetry considerations were used for the case of deadlocks. When considering serializability and equivalence of executions, this is much more subtle - permuting coordinates will certainly permute execution paths, but it is not clear how to choose a representative of such an orbit in a way such that equivalences of executions give rise to equivalences of representatives. 
\acknowledgement{It is a pleasure to thank Roy Meshulam, Martin Raussen and Krzysztof Ziemiansky for very helpful discussions. Moreover, the reviewers have given very helpful references and suggestions and made this a better paper. }
\bibliographystyle{amsplain}
\bibliography{cutoff}

\end{document}